\documentclass[preprint,superscriptaddress,showpacs,nofootinbib,showkeys,tightenlines]{revtex4-1}

\usepackage{amsmath,amssymb,amsthm,mathrsfs,bbm,mathtools}

\usepackage{enumerate}

\usepackage{graphicx}
\usepackage{caption}
\usepackage{subcaption}

\usepackage[pdftex,bookmarks,colorlinks,breaklinks]{hyperref}  
\hypersetup{linkcolor=red,citecolor=blue,filecolor=green,urlcolor=blue}

\usepackage{units}

\usepackage{verbatim,indentfirst}
\usepackage[title,titletoc]{appendix}

\theoremstyle{plain}
\newtheorem{theorem}{Theorem}
\newtheorem*{theorem*}{Theorem}
\newtheorem{lemma}[theorem]{Lemma}
\newtheorem*{lemma*}{Lemma}
\newtheorem{proposition}[theorem]{Proposition}
\newtheorem*{proposition*}{Proposition}

\newtheorem{definition}{Definition}

\makeatletter
\renewcommand{\p@subsection}{}
\renewcommand{\p@subsubsection}{}
\makeatother

\newcommand*{\ketbra}[2]{\lvert #1 \rangle\!\langle #2 \rvert}

\def\tr{{\rm Tr}}

\def\eps{\epsilon}
\def\id{\mathbbm{1}}
\def\Dh{D_{\rm H}}
\def\kb{k_{\rm B}}
\def\Wext{W_{\rm gain}^{\eps}}
\def\Wform{W_{\rm cost}^{\eps}}

\def\comp{+}
\def\gorder{\xmapsto{\beta,\mu}}
\def\eqmaj{\succ_{\beta,\mu}}

\usepackage{color}
\usepackage[dvipsnames]{xcolor}
\usepackage[normalem]{ulem}

\usepackage{tikz}
\usepackage{pgfplots}
\pgfplotsset{compat=1.3}
\pgfplotsset{width=0.5*\textwidth}

\makeatletter
\newcommand\footnoteref[1]{\protected@xdef\@thefnmark{\ref{#1}}\@footnotemark}
\makeatother

\usepackage{soul}

\begin{document}

%
%
\title{Beyond heat baths: Generalized resource theories for small-scale thermodynamics}

%
%
\author{Nicole~Yunger~Halpern\footnote{E-mail: nicoleyh@caltech.edu}}
\affiliation{Institute for Quantum Information and Matter, Caltech, Pasadena, CA 91125, USA}
\affiliation{Perimeter Institute for Theoretical Physics, 31 Caroline Street North, Waterloo, Ontario Canada N2L 2Y5}

\author{Joseph~M.~Renes\footnote{E-mail: renes@phys.ethz.ch}}
\affiliation{Institute for Theoretical Physics, ETH Z\"{u}rich, Switzerland}

\date{\today}



%
%
\begin{abstract}
Thermodynamics has recently been extended to small scales 
with \emph{resource theories} that model heat exchanges.
Real physical systems exchange diverse quantities: heat, particles, angular momentum, etc.
We generalize thermodynamic resource theories 
to exchanges of observables other than heat, to baths other than heat baths, 
and to free energies other than the Helmholtz free energy. These generalizations are illustrated with ``grand-potential'' theories that model movements of heat and particles. Free operations include unitaries that conserve energy and particle number. From this conservation law and from resource-theory principles, the grand-canonical form of the free states is derived. States are shown to form a quasiorder characterized by free operations, $d$-majorization, the hypothesis-testing entropy, and rescaled Lorenz curves. We calculate the work distillable from, and we bound the work cost of creating, a state. These work quantities can differ but converge to the grand potential in the thermodynamic limit. Extending thermodynamic resource theories beyond heat baths, we open diverse realistic systems to modeling with one-shot statistical mechanics. Prospective applications such as electrochemical batteries are hoped to bridge one-shot theory to experiments.
\end{abstract}

\maketitle

%
%
%
%
\section{Introduction}
Advances in small-scale experiments and in quantum information have generated interest in ``thermodynamics without the thermodynamic limit.'' Recent experiments involve molecular motors and ratchets~\cite{LacosteLM08,SerreliLKL07}, optical thermal ratchets~\cite{Faucheux95}, the unfolding of one DNA or RNA molecule~\cite{BustamanteLP05,CollinRJSTB05,LiphardtDSTB02,AlemanyR10}, and nanoscale walkers~\cite{ChengSHEL12}. Analyses of these experiments feature thermodynamic concepts such as heat, work, and equilibrium. These concepts are not well-defined outside the thermodynamic limit of $n \to \infty$ particles. Hence the experimental advances in single-molecule manipulations invite us to extend thermodynamics to small scales.

The resource-theory framework developed in quantum information theory has recently been successfully applied to this problem. Resource theories have been used to calculate how efficiently scarce quantities can be distilled and transformed via cheap, or ``free,'' operations~\cite{CoeckeFS14}. Perhaps the most famous example is the resource theory of pure bipartite entanglement (which we will call ``entanglement theory'')~\cite{HorodeckiHHH09}. In entanglement theory, agents distill Bell pairs of maximally entangled qubits, usable to simulate quantum channels, from partially entangled states via local operations and classical communications (LOCC). Other resource theories quantify the values of asymmetry~\cite{BartlettRS07,marvian_theory_2013,BartlettRST06}, quantum-computation tools~\cite{VeitchMGE14}, and information~\cite{HHOShort,HHOLong,GourMNSYH13}. Benefits of the resource-theory framework include its operational formulation and the explicit modeling of all resources with physical degrees of freedom. 

To an agent with access to a heat bath, nonequilibrium states have value because work can be extracted from them and stored in a battery. Nonequilibrium states' values have been quantified with a family of equivalent resource theories, each associated with an inverse temperature $\beta$ of the bath~\cite{Janzing00,FundLimits2,BrandaoHORS13,BrandaoHNOW13,HorodeckiO13,FaistOR14}. We call these resource theories \emph{Helmholtz theories}, as the central results involve variations on the Helmholtz free energy $F := E - TS$.

Many experiments involve baths other than heat baths, involve interactions other heat exchanges, and are characterized by free energies other than the Helmholtz free energy. The Gibbs free energy 
$G  :=   E - TS  +  pV$ describes processes that occur at fixed temperatures and pressures, such as tabletop chemical reactions. The grand potential $\Phi  :=   E - TS - \mu N$ 
describes heat-and-particle exchanges; and other free energies describe electrochemistry, magnetic fields, mechanical stress and strain, etc.~\cite{Alberty01,Callen85}. 
Different types of baths (equivalently, different types of interactions, or different free energies) invite modeling by different families of resource theories. Each family's constituents correspond to different values of the bath's properties. For example, each member of the family of Helmholtz theories corresponds to one value of the inverse temperature $\beta$. Altogether, the families describing different baths form an extended family of thermodynamic resource theories amenable to experimental investigation in the present or near future. 

We introduce this extended family in this paper, illustrating the formalism with heat-and-particle exchanges. In grand-potential resource theories, free operations conserve energy and particle number. The only states that, if free, prevent such resource theories from being trivial are shown to be grand canonical ensembles
$e^{-\beta ( H - \mu N ) } / {Z}$. We derive the grand canonical ensemble upon establishing rigorously, using the resource-theory formalism, that the free states in Helmholtz theories are canonical ensembles
$e^{- \beta H} / Z$. States are shown to form a quasiorder characterized by a variant of majorization called \emph{$d$-majorization}, which is related to binary hypothesis testing. By exploiting the quasiorder, we calculate the work extractable from, and bound the work required to create, one copy of a state $R = (\rho, H, N)$, even by protocols that have a specified probability of failing. The work yield and work cost are shown to differ from each other in general, unlike in conventional thermodynamics, as observed in~\cite{FundLimits2}. In the limit as the number of copies of $R$ extracted from or created approaches infinity, the average work yield and work cost approach the difference $\Phi(\rho, H, N)  -  \Phi( {\gamma}_\rho, H, N)$ between the state's grand potential and the corresponding equilibrium state's grand potential.

We have structured our results as follows. In the next section, we review the resource-theory framework and define the family of generalized thermodynamic resource theories. We will illustrate the family with grand-potential theories. In Sec.~\ref{section:FreeStates}, we deduce the unique form that free states can assume in these theories. In Sec.~\ref{section:Quasiorder}, we interrelate the quasiorder of states, the generalized notion of majorization, and binary hypothesis testing. In Sec.~\ref{section:OneShotWork}, we define work in the resource-theory framework and use the quasiorder to determine the work yield and work cost of creating single instances of arbitrary states. In Sec.~\ref{section:ManyCopies}, we show that these one-shot work quantities imply asymptotic results similar to traditional thermodynamics.
We conclude by discussing possible applications of our generalized framework 
to real physical systems.

This work bridges the information-theoretic tool 
of thermodynamic resource theories to physical reality.
We pave the way for physical realizations, with experimental platforms, 
of entropic predictions about small scales.

%
%
%
%
\section{Thermodynamic resource theories}
\label{section:DefineTheories}
First, we introduce the resource-theory framework. We define generalized thermodynamic resource theories, then illustrate them with grand-potential theories.

%
%
\subsection{The resource-theory framework}
\label{sec:rtgen}

The resource-theory framework models experimental situations in which some physical transformations between quantum states are difficult, while others are easy~\cite{CoeckeFS14}. In quantum optics, for instance, generating coherent states is easy (e.g., using a laser), while generating number states is difficult. In entanglement theory, classical communication and physical operations on systems possessed by one party each (LOCC) are easy, while quantum operations on systems distributed amongst multiple parties are impossible. The states that are difficult to create can be regarded as resources since, with free operations, they can simulate difficult operations. Given a maximally entangled state, separated parties restricted to LOCC can implement a quantum channel.

A resource theory is defined by physical operations assumed to be easy, or \emph{free}. Free operations include the creation of \emph{free states}; all other states are resources. This definition specifies an ordering of states: States $A$ and $B$ are ordered as $A \mapsto B$ if free operations can create $B$ from $A$. The ordering $\mapsto$ is a quasiorder, satisfying reflexivity ($A\mapsto A$) and transitivity ($A \mapsto B$ and $B \mapsto C$ implies $A \mapsto C$). The quasiorder differs from a partial order: Even if $A\mapsto B$ and $B \mapsto A$, $A$ is not necessarily $B$~\cite{MarshallO79}. 

Functions of resources that respect the quasiorder, in the sense that $A\mapsto B$ implies $f(A)\geq f(B)$, are termed \emph{resource monotones}~\cite{MarshallOA10,BrandaoHNOW13,GourMNSYH13}. Simple sets of monotones completely characterize the quasiorders in some resource theories, including the resource theories in this paper. Many monotones have operational interpretations~\cite{BrandaoHNOW13,GourMNSYH13}.

We are interested only in resource theories whose quasiorders are nontrivial---in which some transformations between some resources are impossible. When independently specifying a resource theory's free operations and free states, we must prevent free states and free operations from being able, together, to generate arbitrary states.

Having introduced the resource-theory framework, we define generalized thermodynamic resource theories and illustrate them with grand-potential theories.

%
%
\subsection{Generalized thermodynamic resource theories}

Requiring that free operations conserve particular physical quantities leads to the extended family of thermodynamic resource theories. Which quantities are conserved depends on which physical systems are modeled, as explained in~\cite{YungerHalpern14}. The Hamiltonian $H$ is conserved in what we have termed \emph{Helmholtz theories}. Janzing \emph{et al.} first defined Helmholtz theories while investigating the resources required to cool systems, though those authors did not use the term ``resource theory''~\cite{Janzing00}. More recently, Brand\~ao \emph{et al.} studied conversions between resources in the asymptotic limit, as the number $n$ of copies of the converted resource diverges ($n \to \infty$)~\cite{BrandaoHORS13}. Horodecki and Oppenheim extended the analysis of conversions beyond the asymptotic limit~\cite{FundLimits2}. The literature about thermodynamic resource theories has  exploded recently:  Since the first draft of the present paper was released, coherences~\cite{LostaglioJR15,LostaglioKJR15} and correlations have been explored~\cite{LostaglioMP14}; connections have been drawn to fluctuation relations~\cite{YungerHalpernDGV14,SalekW15,AlhambraOP15}; and free operations have been generalized~\cite{FaistOR15}.

 In the \emph{grand-potential} theories focused on in this paper, free operations preserve total energy and total particle number. Free states in thermodynamic resource theories such as grand-potential theories model baths such as heat-and-particle reservoirs. As we shall see in Section~\ref{section:FreeStates}, the free states must be equilibrium states, such as grand canonical ensembles if $H$ and $N$ are conserved. If nonequilibrium states are free, the resource theory's quasiorder becomes trivial. 

We associate different baths, (equivalently, different physical quantities that can be exchanged, or, in anticipation of Sec.~\ref{section:ManyCopies}, different free energies) with different families of thermodynamic resource theories. The family of grand-potential theories models heat-and-particle exchanges. The resource theories in each family differ only by the values of the intensive variables that characterize the bath. To specify a grand-potential resource theory, one specifies an inverse temperature $\beta$ and a chemical potential $\mu$. (For simplicity, we focus on systems that contain particles of only one type. To specify a grand-potential theory that models exchanges of particles of $k$ types, one specifies $\beta, \mu_1, \mu_2, \ldots, \mu_k$.) 

Now, let us define the thermodynamic resource theories, their states, and their operations more precisely. Free operations preserve quantities represented, in conventional thermodynamics, by extensive variables. These variables are represented by operators that, with a density operator, define a state. To specify a state $R$ in a grand-potential theory, one specifies a density operator, a Hamiltonian, and a number operator:
\begin{equation}
  R = (\rho,  H,  N).
\end{equation}
These operators are defined on a quantum state space (Hilbert space) $\mathcal H_R$. How a physical system's bath and interactions translate into intensive and extensive variables that define a general family of thermodynamic resource theories detailed in~\cite{YungerHalpern14}.

For simplicity, we specialize to states whose operators commute with each other
($[\rho, H]  =  [\rho, N]  =  [H, N] = 0$) and have discrete, finite spectra.
(Since the initial release of this paper, 
noncommuting operators have been discussed in~\cite{YungerHalpern14,Imperial15,teambristol,TeamBanff}.)
However, we do not restrict the forms of $H$ and $N$ further. We denote the dimension of $\mathcal H_R$ by $d_R$. The density operator $\rho$ can be represented by a matrix diagonal relative to the eigenbasis shared by $H$ and $N$. Such a \emph{quasiclassical} density operator is fully specified by a vector $r$, which we shall call the \emph{state vector}, of its eigenvalues. Hence we also denote the state by $R=(r,H,N)$. The ordering of the elements $r_i$ in $r$ is discussed in Sec.~\ref{section:Quasiorder}.

The composition of the state $R=(\rho, H_R, N_R)$ on $\mathcal H_R$ with the state
$S=(\sigma, H_S, N_S)$ on $\mathcal H_S$ is defined as
\begin{align}
R\comp S=(\rho\otimes \sigma,H_R+H_S,N_R+N_S),
\end{align}
wherein $H_R+H_S=H_R\otimes \id_S+\id_R\otimes H_S$ and $N_R+N_S$ is defined similarly.

As detailed in Section~\ref{section:FreeStates}, free states have density operators whose the probabilities equal Boltzmann factors. The free states in the grand-potential theory take the form 
\begin{align}
\label{eq:boltzweights}
\gamma=e^{-\beta (H-\mu N)}/Z,
\end{align} 
wherein $\beta$ and $\mu$ are real numbers and $Z$ is the normalization factor, or partition function. We denote free states by $G=(\gamma,H,N)$ or $G=(g,H,N)$. Each resource $R=(\rho,H,N)$ is associated with an equilibrium state $G_R=(\gamma_R,H,N)$, or $G=(g_R,H,N)$.

%
%
\subsubsection{Free operations}

We call the free operations in thermodynamic resource theories \emph{equilibrating operations}. They are defined in our grand-potential example as follows.

\begin{definition}[Equilibrating operation]
\label{def:equiop}
In the grand-potential theory defined by $(\beta, \mu)$, an \emph{equilibrating operation} on a state
$R = (\rho, H_R, N_R)$ is any realization of the following three steps:
\begin{enumerate}[(a)]
\item  the drawing of a free state $G  =  (\gamma, H_G, N_G)$ from the bath;
\item  the performing of a unitary transformation $U$ on $R\comp G$, wherein $[U,H_R+H_G]=[U,N_R+N_G]=0$; and
\item  the discarding (tracing out) of any subsystem $A$ associated with its own Hamiltonian and number operator. 
\end{enumerate}
The operation is a completely positive trace-preserving linear map of the form
\begin{equation}
   R  \mapsto  R' =
   ( \tr_A (U [\rho \otimes \gamma] U^\dag ),   H_R + H_G  -  H_A,   N_R + N_G  -  N_A ).
\end{equation}
\end{definition}

Free operations can mix levels whose energies equal each other and whose particle numbers equal each other. We call free operations \emph{equilibrating operations} because (as shown in Sec.~\ref{section:Quasiorder}) free operations monotonically evolve states toward equilibrium states. Equilibrating operations induce on states a quasiorder that we denote by $R\gorder R'$. 

In thermodynamic resource theories other than the grand-potential theories we focus on, free unitaries preserve operators associated with other extensive variables. For example, if a system has $N_1$ particles of Species $1$ and $N_2$ particles of Species $2$, $[U,  N_{1_{\rm tot}}]  =  [U,  N_{2_{\rm tot}}]  =  0$. 

Equilibrating operations idealize the operations easily performable by thermodynamic experimentalists.
Experimentalists cannot perform all unitaries
that preserve energy, particular number, etc.
Thermal operations were ``coarse-grained'' to more realistic operations in~\cite{Perry15}.
We expect similar coarse-graining 
to bridge equilibrating operations from idealization to reality.

%
%
%
\section{Unique form of free states}   \label{section:FreeStates}

The form of equilibrating operations in Definition~\ref{def:equiop} implies that only grand-canonical ensembles can be free in grand-potential theories, else the quasiorder breaks down. The breakdown manifests in two ways.

\begin{theorem}
\label{thm:freestates}
Consider any grand-potential resource theory in which each pair $(H,N)$ corresponds to exactly one free state $G = (\gamma, H, N)$. If $\gamma$ does not have the Boltzmann form of Eq.~\eqref{eq:boltzweights}, then 
\begin{enumerate}[(a)]
\item some resources $R$ can be generated solely with equilibrating operations: $G \gorder R$, and 
\item equilibrating operations can transform one copy of any state $R$ into one copy of any state $S$: 
$R \gorder S$.
\end{enumerate}
\end{theorem}

The proofs appear in Appendix~\ref{section:FreeStateApp}, but we sketch the main ideas here.
First, we derive the form of the free states in Helmholtz theories.\footnote{\label{FreeStateNote}
In~\cite{BrandaoHNOW13}, the canonical form $(e^{ -\beta H } / Z,  H)$ of the free states in Helmholtz theories is argued to follow from~\cite{PuszW78}. According to~\cite{PuszW78}, only canonical ensembles are completely passive: No work can be extracted from canonical ensembles, even from infinitely many, in the absence of other resources. The work extraction in~\cite{PuszW78}, however, is not formulated as in Helmholtz resource theories. How to translate the result from~\cite{PuszW78} into resource theories may not be obvious to all readers. To clarify this subtlety and others, we derive free states' forms directly from the resource-theory framework. Around the time our paper was released, an alternative approach appeared in~\cite{BrandaoHNOW13}.}
Then, we bootstrap from Helmholtz theories to Theorem~\ref{thm:freestates}, which concerns grand-potential theories.

To prove Claim (a), we apply the derivation of the forms of the free states in the resource theory of nonuniformity~\cite{HHOLong}, which models closed isolated systems~\cite{YungerHalpern14}. Consider some energy-and-particle-number eigensubspace $S_{E_i, N_j}$. The free state $\gamma$ has some weight on $S_{E_i, N_j}$. If that weight is distributed nonuniformly across the levels in $S_{E_i, N_j}$, free operations can redistribute the weight arbitrarily across $S_{E_i, N_j}$, generating states not defined as free. Claim (b) follows from modifying an argument by Janzing \emph{et al}.~\cite{Janzing00}. The argument concerns the ``effective temperatures'' of the states that can be created from given resources in the absence of any bath. 

These resource-theory derivations of equilibrium ensembles offer operational alternatives to assumptions such as the Ergodic Hypothesis. According to the Ergodic Hypothesis, uniform distributions represent equilibrated isolated systems' states. Such assumptions have drawn criticism~\cite{JaynesI57,Lenard78}, lending operational replacements appeal. 

%
%
%
%
\section{Quasiorder on states}   \label{section:Quasiorder}

The quasiorder induced by equilibrating operations on quasiclassical states is equivalent to a generalization of majorization. Veinott defined this generalization first, calling it \emph{$d$-majorization}~\cite{Veinott71}. Ruch and collaborators (who called $d$-majorization \emph{the mixing distance})~\cite{RuchSS78,Ruch75} applied $d$-majorization to physics, as did Uhlmann and colleagues~\cite{uhlmann_ordnungsstrukturen_1978,alberti_dissipative_1981}. We dub this quasiorder in general thermodynamic resource theories \emph{equimajorization}, because it is $d$-majorization relative to equilibrium states.

\begin{definition}   \label{definition:EquiMajor}
Let $R$ and $S$ denote states in any grand-potential theory defined by $(\beta, \mu)$. Let $g_R$ and $g_S$ denote the corresponding equilibrium states' state vectors, which contain $d_R$ and $d_S$ elements respectively. ${R}$ \emph{equimajorizes} ${S}$, written as
${R}   \eqmaj   {S}$, if there exists a $d_S\times d_R$ stochastic matrix $M$ such that 
     \begin{equation}   \label{eq:MCondns}
          M r   =   s,
          \quad
          M g_R  =  g_S, 
          \quad {\rm and} \quad
          \sum_{i=1}^{ d_S }  M_{ij}  =  1   \; \; \forall   j = 1, 2, \ldots, d_R.
      \end{equation}
 \end{definition}

In the resource theory of nonuniformity, which models closed isolated systems, equilibrium states are microcanonical ensembles: $g_R  =  \tfrac1{d_R}(1,1,\ldots, 1)$~\cite{HHOLong,GourMNSYH13}. Relative to these uniform states, equimajorization reduces to majorization~\cite{MarshallOA10}.

Janzing \emph{et al.}\ established that the quasiorder on quasiclassical resources is equivalent to equimajorization in Helmholtz theories~\cite[Theorem 5]{Janzing00}. 
An alternative proof appears in~\cite{FundLimits2}. In Appendix~\ref{section:QuasiorderApp}, we extend the proof technique in~\cite{Janzing00} to grand-canonical theories, obtaining the following result.

\begin{theorem}   \label{theorem:FreeOpMajor}
Let $R$ and $S$ denote states in the grand-potential theory defined by $(\beta, \mu)$.
There exists an equilibrating operation that maps $R$ to $S$ if and only if $R$ equimajorizes $S$:
\begin{align}
R\gorder S \Longleftrightarrow R \eqmaj S.
\end{align}
\end{theorem}

As mentioned above, sets of resource monotones completely characterize equimajorization and so characterize the existence of equilibrating operations. One such set consists of the \emph{$f$-divergences}~\cite{csiszar_informationstheoretische_1963,morimoto_markov_1963,ali_general_1966}.
Every convex function $f$ corresponds to an $f$-divergence
\begin{align}
\phi_f(R)=\sum_{i=1}^{d_R}g_{i}\,f\left(\frac{r_i}{g_{i}}\right),
\end{align}
wherein $g_i$ denotes the $i^{\rm th}$ element of the state vector $g_R$ of the equilibrium state associated with $R$. Subsets of the $f$-divergences suffice to characterize equimajorization, as shown below.  Various choices of $f$ lead to well-known functions~\cite{liese_divergences_2006}. For example, $f(x)=x\log x$ and $f(x)=-\log x$ lead to the relative entropies $D(r ||  g_R)=\sum_{i=1}^{d_R}r_j \log (r_i/g_i)$ and $D(g_R||r)$.
The function $f(x)=(x^\alpha-1)/(\alpha-1)$ leads to the R\'{e}nyi divergences 
$D_\alpha(r  ||  g_R)
=   \frac1{1-\alpha}  \log   \left(  \sum_{i=1}^{d_R}r_i^\alpha g_i^{(1-\alpha)}  \right)$ for $\alpha\geq 0$.

The \emph{Lorenz curve}, introduced by Lorenz in economics~\cite{Lorenz05}, encodes another complete set of monotones. Lorenz curves were applied recently to  Helmholtz resource theories~\cite{FundLimits2}. In a grand-potential theory, the rescaled Lorenz curve $L_R:[0,1]\rightarrow [0,1]$ represents the state $R$. The curve is the piecewise linear function that connects the points 
\begin{align}
\label{eq:lorenzpoints}
(t_k,L_R(t_k))=\left\{\begin{array}{ll}
(0,0) & k=0\\
\left(\sum_{j=1}^k g_{\pi(j)}, \sum_{j=1}^k r_{\pi(j)}\right) & k\in\{1,\dots,d_R\}\\
\end{array}\right.,
\end{align} 
wherein $\pi$ denotes a permutation such that the sequence $(r_{\pi(j)}/g_{\pi(j)})_j$ is non-increasing. 
In accordance with~\cite{Lorenz05,FundLimits2}, we define the rescaled Lorenz curve as a monotonically increasing concave function. (Different conventions appear elsewhere~\cite{MarshallOA10}.)

Having defined the $f$-divergences and the rescaled Lorenz curve, we will state their relationship with equimajorization. Ruch, Schranner, and Seligman first proved this relationship for continuous systems~\cite{RuchSS78}, using tools from measure theory. Uhlmann proved the relationship more directly, for discrete systems, which we address~\cite{uhlmann_ordnungsstrukturen_1978}. By following Uhlmann, we will prove this proposition in Appendix~\ref{section:QuasiorderApp}:

\begin{proposition} 
\label{prop:equimajorization}
For any states $R$ and $S$ in the grand-potential resource theory defined by $(\beta, \mu)$, the following are equivalent:
\begin{enumerate}[(a)]
\item $R\eqmaj  S$.
\item $L_R(t)\geq L_S(t)$ for all $t\in [0,1]$.
\item $\phi_{f_a}(R)\geq \phi_{f_a}(S)$ for every function 
$f_a(t)   =   \max  \{0, t-a\}$ associated with any $a \in \mathbbm{R}$.
\item $\phi_f(R)\geq \phi_f(S)$ for all continuous convex functions $f$.
\end{enumerate}
\end{proposition}
\noindent An illustration appears in Fig.~\ref{fig:lorenz}.

\begin{figure}[ht]
\centering
\includegraphics[width=.35\textwidth]{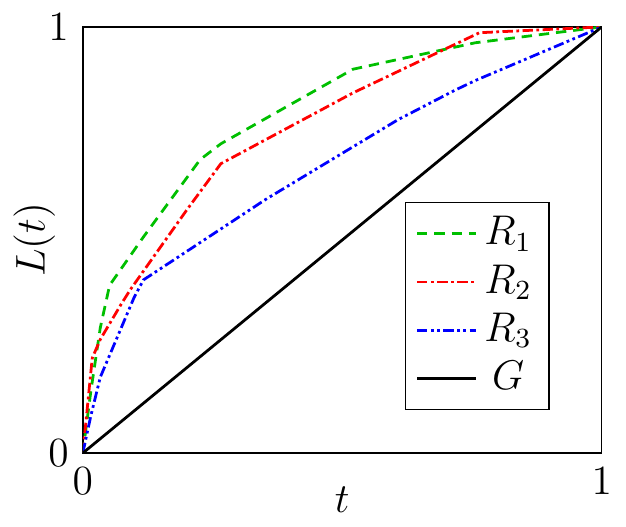}
\caption{\label{fig:lorenz} (Color online)
Rescaled Lorenz curves for three resources ($R_1$, $R_2$, $R_3$) and an equilibrium state ($G$). The Lorenz curve encodes the quasiorder on states, as equilibrating operations can transform $R$ into $R'$ if and only if $L_R(t)\geq L_{R'}(t) \; \: t  \in [0, 1]$. Here, $R_1\gorder  R_3$ and $R_2\gorder  R_3$, but $R_1$ and $R_2$ are incomparable. The equilibrium state, having the linear Lorenz curve $L_G(t)=t$, is at the bottom of the quasiorder.}
\end{figure}

Having characterized $L_R$ in terms of eigenvalues, we explain its relationship with hypothesis testing. The rescaled Lorenz curve is equivalent to the minimal Type II error probability, cast as a function of the Type I error probability, in an asymmetric hypothesis test. Harremo\"es noted the relationship between Lorenz curves and hypothesis tests~\cite{harremoes_new_2004}; we establish the relationship more concretely. 

An \emph{asymmetric hypothesis test} is used to distinguish whether a given state is $\rho$ or $\sigma$. As indicated by our notation, hypothesis testing can be defined in quantum contexts. A test can be thought of as a two-outcome positive operator-valued measurement (POVM) $\{ Q, \id - Q \}$. If the measurement yields the outcome $Q$, the state is likely $\rho$. If the measurement yields $\id - Q$, the state is likely $\sigma$. A Type I error occurs if the state is $\rho$ but $\id - Q$ obtains, so the state seems likely to be $\sigma$. A Type II error occurs if the state is $\sigma$ but seems likely to be $\rho$. The \emph{optimal} test minimizes the Type II error probability while preventing the Type I error probability from exceeding some tolerance $\eps \in [0, 1]$. 

The optimal Type II error probability is
\begin{align}
\label{eq:opttypeIIerror}
b_\epsilon(\rho||\sigma)  :=  \mathop{\min_{\tr[Q\rho] \geq 1 - \eps}}_{0\leq Q\leq \id} \tr[Q\sigma]. 
\end{align}
The condition $\tr[Q\rho] \geq 1 - \eps$ is called the \emph{constraint}, and $\tr[Q\sigma]$ is the \emph{objective function}. Equation~\eqref{eq:opttypeIIerror} defines a semidefinite program, a type of convex optimization, which has a dual form:
\begin{align}  \label{eq:htdual}
   b_\eps(\rho  ||  \sigma)=
   \mathop{\max_{\mu\rho-\sigma\leq \tau}}_{\mu,\tau\geq 0}\,\,
   \left\{   (1 - \eps) \mu-\tr[\tau]   \right\}.
\end{align}
The primal and dual forms' equivalence follows from properties of semidefinite programs~\cite{jensen_generalized_2013}. In quasiclassical notation, $Q$ is represented by a matrix, and traces are replaced by sums.

Hypothesis testing can be related to rescaled Lorenz curves as follows. Consider distinguishing 
between the state vector $r$ in the quasiclassical state $R = (r, H, N)$ and the $g_R$ in the equilibrium state $G_R  =  (g_R, H, N)$.
\begin{lemma}
\label{lem:Lorenzhypotest}
The inverse of $\epsilon \mapsto b_\epsilon(r||g_R)$ is the piecewise linear function that connects the points $(L_R (t_k),1-t_k)$, wherein $t_k$ and $L_R (t_k)$ define the rescaled Lorenz curve for $R$. That is,
\begin{align}
   (t_k,   1  -  L_R (t_k) )
   =   ( b_\epsilon(r||g_R),   \epsilon ).
\end{align}
\end{lemma}
\noindent The proof appears in Appendix~\ref{section:QuasiorderApp}.

%
%
%
%
\section{One-shot work yield and cost}
\label{section:OneShotWork}

Let us quantify the work required to create, and the work extractable from, one copy of a state $R  =  (r, H, N)$ via protocols that can fail, as realistic protocols can. Upon motivating the calculation, we introduce the hypothesis-testing entropy $\Dh^\eps$, incorporate a failure probability into equilibrating operations, and define work in thermodynamic resource theories. Finally, we calculate the extractable work and bound the work cost. Proofs appear in Appendix~\ref{section:FaultyWork}. 

Conventional thermodynamics concerns the average work $\langle W_{\rm gain} \rangle$ extractable from, and the average cost $\langle W_{\rm cost} \rangle$ of creating, states by infallible protocols in the asymptotic limit. In the \emph{asymptotic limit}, or thermodynamic limit, infinitely many identical copies of $R$ are extracted from or created. $\langle W_{\rm gain} \rangle$ and $\langle W_{\rm cost} \rangle$ depend on the Shannon entropy $S_{\rm S}$, itself an average:
\begin{equation}
   S_{\rm S}(r)
   :=   \sum_i   r_i  \ln  r_i
   =    \langle \ln r_i \rangle_r.
\end{equation}

If few copies of a state are extracted from or created, the average cost or yield quantifies the protocol's efficiency poorly. Alternatives to $S_{\rm S}$, called \emph{one-shot entropies}, quantify efficiencies in information-processing (e.g.,~\cite{RennerThesis,tomamichel_leftover_2011,renes_noisy_2011,renes_one-shot_2012,wang_one-shot_2012})
and statistical-mechanics (e.g.,~\cite{DahlstenRRV11,DelRioARDV11,EgloffDRV12,FundLimits2,BrandaoHNOW13,Dahlsten13}) problems
that involve few systems or trials. In addition to involving finite numbers, realistic protocols have nonzero probabilities of failing to accomplish their purposes. Failure probability has been incorporated into one-shot entropies as a parameter $\epsilon$~\cite{RennerThesis,jensen_generalized_2013}.

One alternative to $S_S$ is the \emph{hypothesis-testing entropy} $\Dh^\eps$. $\Dh^\eps$ is defined in terms of the hypothesis test quantified in Eq.~\eqref{eq:opttypeIIerror}. The work extractable from, and the work cost of creating, one copy of a state $R$ will be quantified with $\Dh^\eps$.
\begin{definition} \label{eq:DEpsH}
The \emph{hypothesis-testing relative entropy} between quantum states $\rho$ and $\gamma$ is defined as
\begin{align}
\Dh^\eps(\rho   ||   \gamma) :=  -\ln b_\eps(\rho   ||   \gamma)
\end{align}
or, equivalently, by $b_\eps(\rho   ||   \gamma)=e^{-\Dh^\eps(\rho   ||   \gamma)}$.
\end{definition}
\noindent Let us incorporate failure probability and work into thermodynamic resource theories.

A faulty operation is defined as a transformation whose output approximates the desired output. Operationally, a state $R'$ approximates a state $R$ if no testing procedure consistent with quantum mechanics can reliably distinguish the states. For simplicity, we focus on approximations $R'$ that differ from $R$ only because of their density operators: If $R = (\rho, H, N)$, then $R' = (\rho', H, N)$.  If $R'$ approximates $R$, we write $R' \approx_\eps R$ and say that $R'$ is \emph{$\epsilon$-close} to $R$. (Equivalently, we write $\rho' \approx_\eps \rho$ and say that $\rho'$ is \emph{$\eps$-close} to $\rho$.)
Since density operators' distinguishability is related to the trace distance, we define 
$R'\approx_\epsilon R$ by $\tfrac12\left\|\rho'-\rho\right\|_1\leq \epsilon$. 


%
%
%
We define work in terms of the changing of the energy level occupied by a \emph{battery}. In statistical physics and mechanics, work is defined as an integral along a path in real space or in phase space. Quantum states can follow paths along which work integrals cannot easily be calculated~\cite{Crooks08}. Work on and by quantum systems has been defined more operationally in terms of a ``work bit'' that has a gap $W$~\cite{FundLimits2} and in terms of a weight that stores gravitational potential energy~\cite{SkrzypczykSP13}. We define work similarly to~\cite{SkrzypczykSP13}.

Our battery ${B}$ is any system that has the following qualities: 
(a) The energies in the range accessed by the agent are finely spaced. (b) The battery occupies an energy eigenstate, being a reliable energy reservoir. By ${B}_E$, we denote the battery resource $( \ketbra{E}{E},  H,  N )$. We assume $\beta E\gg 1$, for if $\beta E\approx 1$, agents can use the battery's equilibrium state to drive processes that require energy $E$. Such a use would contradict our physical notion of useful work. 

Having defined $B$, we can define the work extractable from, and the work cost of, a state transformation. If the battery transitions from $B_{E_i}$ to $B_{E_f}$ while $R$ transforms into $S$, the transformation outputs the work $E_f - E_i$ (which is negative if the transformation costs work). 
The $\epsilon$-work value of a resource ${R}$ is defined as the greatest $W$ for which  
\begin{align}   \label{eq:workvalue}
   R\comp  {B}_E   \gorder_\epsilon {B}_{E+W}.
\end{align}
The work cost of $\epsilon$-approximately creating $R$ is the least $W$ such that
\begin{align}   \label{eq:workcost}
   {B}_{E+W}   \gorder_\epsilon R\comp {B}_E.
\end{align}
Formally,
\begin{align}
   \Wext (R)   &=   \max   \{W   :   R  \comp   {B}_E    
      \gorder_\epsilon 
      {B}_{E+W},   \:   \beta E\gg 1\}\\
   \Wform (R)   &=   \min   \{W   :   {B}_{E+W}   
      \gorder_\epsilon    
      {R}   \comp   {B}_E,   \:   \beta E\gg 1   \}.
\end{align}

The simplicity of our battery model facilitates calculations.
Realistic features could be incorporated as follows.
First, the battery could occupy a mixed state, 
or a superposition of energy eigenstates, at any stage in either protocol.
Second, the system and battery could begin or become entangled.
Such entanglement could be analyzed as in~\cite{Frenzel14}.
Frenzel \emph{et al.} point out that a classical field 
is often assumed to raise and lower a quantum system's energy levels.
But fields are not classical and become entangled with the system.
A battery might become entangled similarly.

Having defined work, we state the work value, and bound the work cost, of $R$. 

\begin{theorem}
\label{thm:extractMain} 
The $\epsilon$-work value of a state $R = (r, H, N)$ associated with the free state
$G_R = (g_R, H, N)$ is
\begin{align}  \label{eq:WExt}
\Wext(R)   =   \tfrac1\beta\Dh^{\eps}(r   ||   g_R).
\end{align}
The work cost of creating an $\epsilon$-approximation to a state $R$ is bounded by 
\begin{align}  \label{eq:WForm}
 \max_{\delta\in (0,1-\eps]}  \left[
     \tfrac1\beta \Dh^{1 - \eps - \delta}(r   ||   g_R)-\tfrac1\beta\log  \left( \tfrac1\delta \right)   \right]
   \leq \Wform(R)
&\leq \tfrac1\beta\Dh^{1 - \eps}(r   ||   g_R)-\tfrac1\beta\log  \left( \tfrac {1-\eps}{\eps} \right).
\end{align}
\end{theorem}

A proof appears in Appendix~\ref{section:FaultyWork}.
Each expression in the theorem contains an entropy $D^\varepsilon_{\rm H} (r || g_R)$, for some error probability $\varepsilon$. The factor $\frac{1}{\beta}$ introduces dimensions of energy. Each bound contains a logarithmic correction.

Theorem~\ref{thm:extractMain} bounds optimal efficiencies.
Thermodynamic optima tend to characterize 
physically unrealizable processes.
Example processes include quasistatic, or infinitely slow, evolutions.
Experiments cannot proceed infinitely slowly.
What implications can Theorem~\ref{thm:extractMain} have
for real physical processes?
As a process is performed increasingly slowly,
its efficiency is expected approach our predictions.
A similar approach has been reported in~\cite{Koski14}.
Koski \emph{et al.} erased a bit of information repeatedly.
As the erasure's speed dropped,
the amount of heat dissipated dropped 
to near the Landauer limit.

%
%
%
%
\section{Work yield and cost of many copies of a resource}  \label{section:ManyCopies}

From the previous section's one-shot work quantities, we can recover results reminiscent of traditional thermodynamics and can compare how $\Wext$ differs from $\Wform$ as the thermodynamic limit is approached. We denote $n$ copies of $R = (r, H, N)$ by $R^{\otimes n}=(r^{\otimes n},\sum_{i=1}^n H_i,\sum_{i=1}^n N_i)$. In the thermodynamic limit, or asymptotic limit, $n \to \infty$. Also in the limit, we show, $\Wext(R^{\otimes n})$ and $\Wform(R^{\otimes n})$ tend to the difference between the grand potential of $R$ and the grand potential of the associated equilibrium state $G_R$. As the thermodynamic limit is approached, $\Wext(R^{\otimes n})$ and $\Wform(R^{\otimes n})$ differ by terms of order $\sqrt n$. 

%
%

To derive the thermodynamic limits of Eq.~(\ref{eq:WExt}) and Ineqs.~(\ref{eq:WForm}), we invoke the Asymptotic Equipartition Property of $D^\epsilon_{\rm H}$~\cite{jensen_generalized_2013}:
\begin{equation}  \label{eq:AEP}
   \lim_{n \to \infty}   \frac{1}{n}
   D^\epsilon_{\rm H} ( r^{\otimes n}  ||  s^{\otimes n} )
   =   D( r || s )
   \quad \forall \epsilon \in (0, 1),
\end{equation}
wherein $r = (r_1, r_2, \ldots, r_d)$ and 
$s = (s_1, s_2, \ldots, s_d)$ denote probability distributions over the same alphabet.
We have used the definition
\begin{equation}
   D(r || s)   :=   \sum_{i=1}^d   r_i   \ln  \left(   \frac{ r_i }{ s_i }   \right)
\end{equation}
of the relative entropy, defining $0 \ln 0 = 0$~\cite{CoverT12}.

Applying Eq.~(\ref{eq:AEP}) to Eq.~(\ref{eq:WExt}) and to both sides of Ineqs.~(\ref{eq:WForm}) yields
\begin{equation}   \label{eq:Asymp1}
   \lim_{n \to \infty} \frac{1}{n}
      \Wext (R^{\otimes n} )
   =   \lim_{n \to \infty} \frac{1}{n}
        \Wform (R^{\otimes n} )
   =   \frac{1}{\beta}  D( r || g_R ).
\end{equation}
In the asymptotic limit, the bounds in Ineqs.~(\ref{eq:WForm}) converge. All strategies of work extraction and state formation, from risky ($\eps\approx 1$) to conservative ($\eps\approx 0$), become equivalent. 

To understand Eq.~(\ref{eq:Asymp1}) further, we invoke the definition of the relative entropy:
\begin{align}
   \frac{1}{\beta}   D(r||g_R) 
   &=  \frac{1}{\beta}   \sum_i r_i \left(\ln r_i - \ln \frac{e^{-\beta(E_i-\mu n_i)} }{Z}\right)\\
   &=    \langle H\rangle_r   - T  \cdot \kb  S_{\rm S}(r)   -   \mu \langle N\rangle_r +  \kb T \ln Z\\
   & =  \Phi_{\beta,\mu} (R)   -    \Phi_{\beta,\mu} (G_R).
\end{align}
Recall that $\Phi := E - TS - \mu N$ denotes the grand potential, and $- \kb T \ln Z$ denotes the equilibrium state's free energy, in conventional thermodynamics. Using one-shot information theory, we have recovered the convergence, in the asymptotic limit, of a state's average work cost and average work yield to a difference between free energies. 

Equation~(\ref{eq:Asymp1}) implies that all resources can be reversibly converted into one another in the asymptotic limit. For any states $R$ and $S$ and for fixed $\epsilon$, there exists an $n$great enough that $R^{\otimes n}\gorder_\epsilon S^{\otimes {m_n}}$ for some $m_n \geq 1$. To create $m_n$ copies of $S$ from $n$ copies of $R$, one extracts all the work possible from $R^{\otimes n}$, then constructs 
$(S^{\otimes m_n})'  \approx_\eps  S^{\otimes m_n}$ from the work. We define the optimal asymptotic conversion rate $\mathsf{R}(R \mapsto S)$ as the asymptotic limit of the supremum of the rates $m_n / n$ achievable by conversion protocols that approximate the desired output arbitrarily well in the asymptotic limit (protocols for which $\epsilon\rightarrow 0$). This rate is
\begin{align}
\mathsf{R}(R \mapsto S) &= \frac{D(r||g_R)}{D(s||g_S)}.
\end{align}

Thus, all nonequilibrium states can be reversibly converted into each other in the asymptotic limit. 
This result may be surprising. One might have thought that resourcefulness can be ``locked'' into one form---energy, particle number, or information---preventing an $R$ whose resourcefulness manifests in energy from transforming into an $S$ whose resourcefulness manifests in particle number. Apparently, such locking does not occur. This asymptotic reversible convertibility resembles that in Helmholtz theories~\cite{BrandaoHORS13} and in the nonuniformity theory~\cite{HHOShort,HHOLong}. Asymptotic reversible convertibility in general resource theories is discussed in~\cite{gour_measuring_2009}.

The Asymptotic Equipartition Theorem dictates the leading order (order-$n$) behavior of $D^\epsilon_{\rm H}$; 
a more-refined analysis reveals the next-leading-order terms. Applying techniques from information theory, we show in Appendix~\ref{section:SecondOrderAsymp} that the latter terms are of order $\sqrt n$: 
\begin{align}
\Wext (R^{\otimes n} ) &=   \frac{1}{\beta} [ nD(r||g_R)  -  O(\sqrt n) ], \qquad \text{and}\\
\Wform (R^{\otimes n} ) &=    \frac{1}{\beta} [ nD(r||g_R)+O(\sqrt n) ].
\end{align}
As one might expect, the work cost $\Wform(R^{\otimes n})$ lies above the thermodynamic value, whereas 
$\Wext (R^{\otimes n})$ lies below. This discrepancy contrasts with conventional thermodynamics, according to which a reversible cycle can extract work from $R$ and use that work to recreate $R$. Outside the thermodynamic limit, such reversible cycles are impossible. The resource-theory framework refines the Second Law of Thermodynamics, as discussed in~\cite{BrandaoHNOW13}.

%
%
%
%
\section{Conclusions}

We have extended the resource-theory formulation of thermodynamics beyond heat baths. 
Earlier thermodynamic resource theories model heat exchanges;
but many physical systems exchange heat, particles, volume, magnetization, and other observables.
We model these exchanges with a set of families of thermodynamic resource theories.
Each family corresponds to one one free energy, one type of interaction, and one type of bath. 

To illustrate mathematical results, we focused on grand-potential theories, 
whose free operations conserve energy and particle number. 
We showed, using resource-theory principles, 
why free states must be grand canonical ensembles. We characterized the quasiorder on states by an extension of majorization, here termed \emph{equimajorization}. We showed that equimajorization can be formulated in terms of rescaled Lorenz curves and of the optimal error probability in asymmetric hypothesis testing. 
The hypothesis-testing entropy was shown to be proportional to the amount of work extractable from a state $R$ and to bound the work cost of creating $R$.  In the asymptotic limit $n \to \infty$, 
$\Wext (R^{\otimes n})$ and $\Wform (R^{\otimes n})$ were shown to converge to a difference
$\Phi (R)  -  \Phi (G_R)$ between grand potentials. 
The convergence rates were shown to differ on the order of $\sqrt{n}$.
In the limit, all states were shown to be reversibly interconvertible.

Opportunities for bringing these resource theories closer to experiments remain.
Examples include the finite sizes of heat baths,
catalysts (ancillas that facilitate transformations
while suffering no or little degradation), 
limitations on how much of a resource can be exchanged,
and the speeds with which transformations can be implemented.
In the presence of an infinitely large heat bath
and enough work,
every resource $R$ can be converted into every other.
As the number $n$ of copies $R$ grows large,
the size of the required bath scales only superlinearly with $n$~\cite{BrandaoHORS13}. 
A finite-sized bath limits the resource hierarchy, 
possibly spoiling the interconvertibility of all resource states. 
The effects of the bath's finiteness
might be incorporated as in~\cite{reeb_improved_2014},
which concerns the cost of erasing 
with a small bath (albeit outside the resource-theory framework).
Another open question concerns
how finite-sized catalysts affect the work cost of resource interconversion~\cite{BrandaoHNOW13}. 
Finally, the optimal efficiencies in the present paper
might characterize only quasistatic---infinitely slow---protocols.
Realistic protocols proceed at finite rates.
Extensions of our results to finite speeds may draw inspiration from~\cite{browne_guaranteed_2014,masanes_derivation_2014}.
As noted in the final reference, finite speeds relate directly to the density of accessible bath levels. 

Generalized thermodynamic resource theories 
open a host of realistic thermodynamic systems 
to modeling with resource theories. 
Particular physical platforms call out for modeling:
heat-and-particle exchanges, electrochemical batteries, chemical reactions, etc.
As thermodynamic potentials other than 
the Helmholtz free energy $F$ characterize common experiments, 
generalized thermodynamic resource theories offer opportunities 
for realizing one-shot statistical mechanics experimentally.

%
%
%
%
\section*{Acknowledgements}
NYH is grateful for conversations with Tobias~Fritz, Iman~Marvian, Markus~M\"{u}ller, Brian~Space, and Rob~Spekkens. JMR acknowledges helpful conversations with Michael Walter. This work was supported by a Virginia Gilloon Fellowship, an IQIM Fellowship, NSF grant PHY-0803371, the Perimeter Institute for Theoretical Physics, the Swiss National Science Foundation (through the National Centre of Competence in Research ‘Quantum Science and Technology’ and grant No. 200020-135048), and the European Research Council (grant No. 258932). The Institute for Quantum Information and Matter (IQIM) is an NSF Physics Frontiers Center that receives support from the Gordon and Betty Moore Foundation. Research at the Perimeter Institute is supported by the Government of Canada through Industry Canada and by the Province of Ontario through the Ministry of Research and Innovation. 
NYH is grateful to Renato Renner for hospitality at ETH Z\"{u}rich during the development of this paper.

\section{Appendices}

%
%
%
%
\begin{appendices}

Below, we prove claims, presented above, about grand-potential theories. We derive the grand-canonical forms of free state vectors; describe the quasiorder; calculate the work $\Wext$ extractable from, and bound the work $\Wform$ required to create, one copy of a state; and show, via second-order asymptotics, that $\Wext$ does not always equal $\Wform$.

%
%
%
%
\section{Derivations of free states' forms}  \label{section:FreeStateApp}

\subsection{Proof of Theorem~\ref{thm:freestates}(a)}
The state vectors of the free states in grand-potential theories are shown to be grand canonical ensembles. We first review the derivation, in~\cite{HHOLong}, of the forms of the free states in the resource theory of nonuniformity, which models closed isolated systems~\cite{YungerHalpern14}. From the nonuniformity result, we deduce the canonical form of the free states in Helmholtz theories. From the Helmholtz result, we bootstrap to grand-potential theories.\footnoteref{FreeStateNote} 

Free operations in the nonuniformity theory are called \emph{noisy operations}. Each noisy operation consists of three steps: Any free state $u$ (whose form is to be derived) can be created, any permutation $\pi$ can be implemented, and any subsystem $A$ can be discarded (marginalized)~\cite{HHOShort,HHOLong,GourMNSYH13}:
\begin{equation}
   r   \mapsto   \sum_A  \pi(  r  \otimes  u ).
\end{equation} 
Resource states are defined here as states that are not free $u$'s (or, in general thermodynamic resource theories, states that are not free $G$'s) that appear explicitly in the definition of free operations. A resource theory is \emph{trivial} if its free operations alone can generate resource states.

As shown in~\cite{HHOLong}, the free states must be uniform probability distributions, lest the nonuniformity theory be trivial. (Indeed, the quasiorder of resources becomes trivial: From enough copies of any state, free operations can generate any other state.) The free states' form is derived as follows~\cite{HHOLong}: Suppose that some nonuniform state $u_0$ is free. By Shannon compressing many copies of $u_0$~\cite{NielsenC10}, agents can create pure states $(1, 0, 0, \ldots, 0)$ for free. Via noisy operations, agents can create noise for free. Able to generate purity and noise, free operations can generate arbitrary states. Only if all free states are uniform is the nonuniformity theory nontrivial.

General thermodynamic resource theories contain the nonuniformity theory as a special case. In grand-potential theories, for example, free operations can arbitrarily permute levels within each sector $S_{E,N}$ that corresponds to one energy $E$ and one particle number $N$. Hence the weight that each free state has on a sector $S_{E, N}$ is distributed uniformly across the levels in $S_{E, N}$. We call this uniformity the \emph{uniform-eigensubspace condition}. The condition is defined in Helmholtz theories as follows.

\begin{definition}
Let $R = (r, H)$ denote a state in any Helmholtz theory, wherein $r = (r_1, \ldots, r_d)$. $R$ obeys the \emph{uniform-eigensubspace condition} if, for every degenerate eigenvalue $E$ of $H$, all the $r_i$ associated with $E$ equal each other. 
\end{definition}

Let us apply the nonuniformity-theory argument to the uniform-eigensubspace condition. 

\begin{proposition}
The free states in each thermodynamic resource theory obey the uniform-eigensubspace condition. If the free states disobeyed the condition, there would exist resources $R$ that equilibrating operations alone could generate:  $G \mapsto  R$.
\end{proposition}
\noindent Suppose that free states disobeyed the uniform-eigensubspace condition. Each free state's state vector $g$ would have some weight $p$ on each sector $S$ that corresponds to some energy, some particle number, etc. Equilibrating operations could distribute $p$ arbitrarily across the levels in $S$ but could not change the value of $p$.

The uniform-eigensubspace condition implies the following three lemmas, which complete our derivation of the canonical form of the free states in Helmholtz theories.

\begin{lemma}   \label{lemma:FreeStateFrac}
Let $H_1$ and $H_2$ denote any Hamiltonians that share an energy gap $\Delta$ (and whose spectra are discrete). Let $E_1$ and $E_1+\Delta$ denote eigenvalues of $H_1$, and let $E_2$ and $E_2+\Delta$ denote eigenvalues of $H_2$. 
Define $G_1=(g_1,H_1)$ as a Helmholtz-theory state whose weights on $E_1$ and $E_1 + \Delta$ are $g_1(E_1)$ and $g_1(E_1+\Delta)$.
Define $G_2=(g_2,H_2)$, $g_2(E_2)$, and $g_2(E_2 + \Delta)$ analogously.
If $G_1\comp G_2$ satisfies the uniform eigensubspace condition, the ratio of the weights depends only on the gap: 
\begin{equation}   \label{eq:GapRatio}
   \frac{  g_1(E_1 + \Delta)  }{  g_1(E_1)  }
   =  \frac{  g_2(E_2 + \Delta)  }{  g_2(E_2)  }.
\end{equation}
\end{lemma}

\begin{proof}  
The eigenenergy $E_1 + E_2 + \Delta$ of $G_1 + G_2$ has a twofold degeneracy. Since $G_1\comp G_2$ satisfies the uniform eigensubspace condition, the weight of 
$g_1 \otimes g_2$ on one degenerate level equals the weight on the other:
\begin{equation}   \label{eq:Fractions}
   g_1(E_1)   g_2( E_2 + \Delta )
   =   g_1(E_1 + \Delta)   g_2 (E_2).
\end{equation}
Since $E_1$ and $E_2$  are arbitrary, each ratio of probabilities depends only on $\Delta$; the other details of $H_1$ and $H_2$ are irrelevant.  
\end{proof}

\begin{lemma}  \label{lemma:CanonLemma2}
Let $G=(g,H)$ denote any free Helmholtz-theory state that has weights $g(E)$ and $g(E + \Delta)$ on either side of an energy gap $\Delta$. The ratio of the weights  varies exponentially with the gap:
\begin{equation}   \label{eq:RatioExp}
   \frac{  g(E + \Delta)   }{   g(E)  }
   =   e^{ - \beta \Delta },
\end{equation}
wherein $\beta \in \mathbbm{R}$.
\end{lemma}

\begin{proof}
Consider a state $G=( g,  H )$ that has three energies separated by gaps $\Delta$:
\begin{equation}  \label{eq:CanonEx}
   H =   E \ketbra{1}{1}   +   (E + \Delta)  \ketbra{2}{2}   
           +   (E + 2 \Delta)  \ketbra{3}{3}   +   \sum_{i=4}^d E_i \ketbra{i}{i}.
\end{equation}
Let us apply Lemma~\ref{lemma:FreeStateFrac} to two copies of $G$, defining $E_1=E$ and $E_2=E+\Delta$. Equation~\eqref{eq:GapRatio} becomes
\begin{equation}   \label{eq:fLog1}
     \frac{   g(E + 2 \Delta)   }{   g(E + \Delta)   }=
   \frac{   g(E + \Delta)   }{   g(E)   }  
   =:f(\Delta),
\end{equation}
which implies
\begin{equation}
   [ f (\Delta) ]^2
   =  \frac{  g(E + 2 \Delta)   }{   g(E +  \Delta)   }
        \frac{  g(E  +  \Delta)  }{  g(E)  }
   =   \frac{   g(E + 2 \Delta)   }{   g(E)   }
   =   f ( 2 \Delta ).
\end{equation}
This scaling implies that, if $f$ is continuous, it is an exponential:
\begin{equation}   \label{eq:fLog2}
   f ( \Delta )   =   e^{ - \beta \Delta }
\end{equation}
for some $\beta \in \mathbbm{R}$.
(The realness of $\beta$ follows from the probabilities' realness.) 

Let $G'$ denote any free state that has the same gap $\Delta$ as $G$.  The composition $G + G'$ obeys the uniform-eigensubspace condition. Hence $G'$ satisfies Eq.~\eqref{eq:fLog2}, by Lemma~\ref{lemma:FreeStateFrac}, even if $G'$ does not have the form in Eq.~\eqref{eq:CanonEx}.
\end{proof}

\begin{lemma}  \label{lemma:CanonLemma3}
All the gaps in all the free states in any given Helmholtz theory correspond to the same $\beta$.
\end{lemma}

\begin{proof}

Let $G = (g, H)$ and $G' = (g', H')$ denote free states in some Helmholtz theory. Let $\Delta$ denote a gap in $H$; and let $\Delta'  =  n \Delta$, wherein $n$ denotes a positive integer, denote a gap in $H'$. We wish to show that $g$ and $g'$ correspond to the same $\beta$. The proof will be extended to rational proportionality constants, then to arbitrary constants.

Some free state $G'' = (g'', H'')$ has $p > n$ equally spaced levels $E, E+ \Delta, \ldots, E + p \Delta$. For example, $G''$ might denote a harmonic oscillator. This $G''$ will serve as a thermometer that interrelates the temperatures of $G$ and $G'$.
Let $E_1 = E$ and $E_2 = E + m \Delta$ for any $m \in \{1, 2, \ldots, p\}$.
By an argument like the one used to prove Lemma~\ref{lemma:CanonLemma2}, 
\begin{align}
   f( m \Delta )  & :=   \frac{ g''( E + m \Delta) }{ g''(E) }  \\
      & =  \left[  \frac{ g''(E + \Delta) }{ g''(E) }   \right]^m  \\
      & =:  f( \Delta )^m,  \label{eq:FreeStateHelp0}
\end{align}
wherein $g''(E + k \Delta)$ denotes the weight on level $k$.

Consider substituting $m = n$ and $\Delta  =  \frac{1}{n} \Delta'$ into the left-hand side (LHS) of,
$f(m \Delta)  =  f( \Delta )^m$:
\begin{align}  \label{eq:FreeStateHelp}
   f( \Delta' )   =   f \left( \Delta \right)^n.
\end{align}
We have related the ratio of the weights across the gap of $G'$ to the ratio of the weights across the gap of $G$. That is, we have related the temperature of $G'$ to the temperature of $G$.

Now, suppose that $\Delta' = \frac{m}{n} \Delta$.  Consider substituting $m \Delta = n \Delta'$ into the LHS of \mbox{$f(m \Delta)  =  f( \Delta )^m$}:
$f(n \Delta')  =  f (\Delta)^m$.
This equation's LHS also equals $f(\Delta')^n$, by Eq.~\eqref{eq:FreeStateHelp0}. Equating the two expressions for $f(n \Delta')$ yields
\begin{align}  \label{eq:FreeStateHelp2}
   f( \Delta' )   =   f( \Delta )^{m / n}.
\end{align}
We have related the temperature of $G'$ to the temperature of $G$, effectively by considering multiple copies of each state.

Finally, suppose that $\Delta' = \alpha \Delta$, wherein $\alpha$ denotes an irrational number. $\alpha$ can be approximated arbitrarily well by a ratio $m/n$. Arbitrarily many copies of $G$ and $G'$ relate the temperature of $G$ to that of $G'$ via Eq.~\eqref{eq:FreeStateHelp2}.
\end{proof}

Lemmas~\ref{lemma:FreeStateFrac}-\ref{lemma:CanonLemma3}, with the normalization condition, imply that the free states in Helmholtz theories are canonical ensembles. This result will facilitate our proof of Theorem~\ref{thm:freestates} about grand-potential theories.

\begin{theorem*}
Consider any grand-potential resource theory in which each pair $(H,N)$ corresponds to exactly one free state $G = (g, H, N)$. If $g$ is not a grand canonical ensemble, some resources $R$ can be generated solely with equilibrating operations: $G \gorder R$.
\end{theorem*}

\begin{proof}
First, we show that each element of $g$ has the form
$e^{- \beta( E_i ) E_i   +   \alpha( n_j ) n_j } / Z$, wherein $\beta(E_i)$ and $\alpha(n_j)$ denote functions of the energy and particle number.
Second, by comparing the grand-potential theory with Helmholtz theories, we will show that $\beta$ and $\alpha$ are constant functions.

Consider the most general state vector associated with $H$ and $N$. Each element has the form 
$e^{ - \beta ( E_i ) E_i   +  \alpha ( n_j ) n_j   +  f ( E_i, n_j ) } / Z$,
wherein $f$ denotes some function and $Z$ normalizes the state.
Recall that every Helmholtz-theory problem can be decomposed into single-energy lemmas that are equivalent to nonuniformity-theory problems. Likewise, every grand-potential--theory problem can be decomposed into lemmas that feature just one $n_j$ apiece and that are equivalent to Helmholtz-theory problems. Therefore, the elements of $g$ that correspond to the same $n$ must form a canonical ensemble. (Rather, they would form a canonical ensemble if normalized appropriately.) These $g$ elements could not form a canonical ensemble if $f$ depended on energy nontrivially. Hence $f(E_i, n_j)  =  f( n_j )$. By an analogous argument, $f$ cannot depend on $n_j$. Hence $f$ is a constant, and each element of 
$g$ has the form   $e^{- \beta(E_i) E_i   +   \alpha(n_i) n_i } / Z$.

$G$ could feature in a problem in which every number operator is trivial: $N = 0$. Such a problem is equivalent to a problem in a Helmholtz theory. In each Helmholtz theory, all free states share a $\beta$ that is a constant function of energy. This $\beta$ must characterize the grand-potential theory. Analogously, all free states in the grand-potential theory share an $\alpha \in \mathbbm{R}$. Hence each element of $g$ has the form $e^{ - \beta E_i + \alpha n_j } / Z$. Define $\mu \in \mathbbm{R}$ such that $\alpha = - \beta \mu$.
\end{proof}

\subsection{Proof Sketch of Theorem~\ref{thm:freestates}(b)}

Here we provide a proof sketch for Theorem~\ref{thm:freestates}(b). For simplicity, we consider a Helmholtz-theory context. We must show that, if the free states do not have the Boltzmann form, the quasiorder on states is trivial: Any state can be created from any other by equilibrating operations. We follow an argument by Janzing \emph{et al.}~\cite{Janzing00} about the effective temperatures present in many copies of a non-Boltzmann state. These effective temperatures can be used to cool a qubit. 

Consider the transformation of $R=(r,H)$ into $S=(s,H)$. It suffices to operate successively on pairs of levels, such as the first and the $j^{\rm th}$. An equilibrating operation of the following form converts $r_j/r_1$ into $s_j/s_1$. 
Let $E_j$ denote the gap between the two levels. Following~\cite[Sec.\ 3]{Janzing00}, we suppose that we have access to a free state on three levels, separated by gaps $E_j$, the probabilities on which are not Boltzmann-weighted. As shown by Janzing \emph{et al.}, the product of $n\rightarrow \infty$ copies of this free state contains pairs of levels, separated by $E_j$, characterized by essentially any desired relative probability $p_j/p_1$.  Two levels of the resource-and-free-state composite are degenerate. One degenerate level is the product of the resource's lower level and the free state's upper level; the other degenerate level consists of the reverse. Swapping the degenerate levels is an equilibrating operation. Easy calculation shows that the swap transforms the resource's probabilities into
\begin{align}
{r'_j}=r_j+\delta
\quad {\rm and} \quad
r'_1=r_1-\delta,
\quad\text{wherein}\quad 
\delta=r_1 p_j-r_j p_1.
\end{align}
If $p_j+p_1=1$, any relative probability $r'_j/r_1$ could be reached by appropriate choice of $p_j/p_1$. Generally, $p_j+p_1$ is very small (exponentially small in $n$ in the Janzing \emph{et al.} example). Hence the relative weights of the levels in $R$ can be changed by only a tiny amount. Repeating the procedure sufficiently many times, however, yields the desired relative probability $s_j/s_1$. Unless the free state has Boltzmann weights, therefore, the quasiorder induced by equilibrating operations is trivial.

%
%
%
%
\section{Quasiorder proofs} \label{section:QuasiorderApp}

Let us prove three statements, introduced in Sec.~\ref{section:Quasiorder}, about the quasiorder on states: Theorem~\ref{theorem:FreeOpMajor}, Proposition~\ref{prop:equimajorization}, and Lemma~\ref{lem:Lorenzhypotest}.

%
%
\begin{theorem*}[\ref{theorem:FreeOpMajor}]
Let $R$ and $S$ denote states in the grand-potential theory defined by $(\beta, \mu)$.
There exists an equilibrating operation that maps $R$ to $S$ if and only if $R$ equimajorizes $S$:
\begin{align}
R\gorder S \Longleftrightarrow R \eqmaj S.
\end{align}
\end{theorem*}

\begin{proof}

This proof is adapted from the proof of~\cite[Theorem 5]{Janzing00}, a Helmholtz-theory analog of our Theorem~\ref{theorem:FreeOpMajor}.\footnote{
An alternative approach to the Helmholtz-theory analog appears in~\cite{FundLimits2}.}
Let $R = (r, H_R, N_R)$ and $S = (s, H_S, N_S)$. 
By $d_R$ and $d_S$, we denote the numbers of elements in $r$ and $s$. 
By $G_R  =  (g_R,  H_R,  N_R)$ and $G_S  =  (g_S, H_S,  N_S)$, we denote the equilibrium states associated with $R$ and $S$. We begin with the easier part of the proof, showing that the existence of an equilibrating operation implies equimajorization.

Assume that some equilibrating operation maps 
$R$ to $S$:
\begin{equation}
   R \gorder
   ( \mathcal{E} ( R ),   H_S,  N_S )
   = (S, H_S, N_S).
\end{equation}
Let $v = (v_1, \ldots, v_{d_R} )$ denote any vector that contains $d_R$ elements.
The $d_R \times d_S$ matrix $M$ that implements $\mathcal{E}$ can be defined by
\begin{equation}
   M v 
   =  \mathcal{E} ( v ).
\end{equation}

By the definition of $\mathcal{E}$, $M r = s$. Since equilibrating operations map equilibrium states to equilibrium states, $M g_R  =  g_S$. We can see as follows that $M$ is stochastic: If $v$ represents a (normalized) state, $\sum_i v_i = 1$. Equilibrating operations preserve normalization, so 
\begin{equation}
   1  =  \sum_{i=1}^{ d_R }   \left[ \mathcal{E} (v) \right]_i
    =   \sum_{i=1}^{ d_R }  [ Mv ]_i,
\end{equation}
wherein $[ w ]_i$ denotes the $i^{\rm th}$ element of any vector $w$.
Mapping normalized vectors to normalized vectors, $M$ is stochastic.
By Definition~\ref{definition:EquiMajor}, $R \succ_{\beta, \mu} S$.

Now, we proceed to the converse claim.
%
%
Assume that $R  \succ_{\beta, \mu}  S$. One can prove that some equilibrating operation maps $R$ to $S$ by augmenting three lines in the proof of Theorem~5 by Janzing \emph{et al.}~\cite{Janzing00}. 
After outlining the latter proof, we explain how to augment it. 

Janzing \emph{et al.} define a particular energy-preserving transformation implemented with a heat bath; consider the limit as the bath's size approaches infinity; and show that, in the limit, the transformation converts the initial state into the equimajorized state.
Blending their notation with ours, we denote the initial state by $(p, H_p)$, the final state by
$(\tilde{p},  H_{ \tilde{p} } )$, and the associated equilibrium states by 
$(g, H_p)$ and $( \tilde{g},  H_{\tilde{p}} )$. 
The Hamiltonian $H_p$ has $l$ levels, and $H_{ \tilde{p} }$ has $\tilde{l}$ levels.

Janzing \emph{et al.} consider the set $\mathcal{S}_n$ of pure eigenstates of 
$H_p  +  H_p^{  n}  +  H_{ \tilde{p} }^{ n }  +  H_{ \tilde{p} }$, wherein $H^n$ denotes $n$ copies of $H$. Each state in $\mathcal{S}_n$ is characterized by two length-$(n + 1)$ strings. Each letter in the first (second) string is a number between $1$ and $l$ ($\tilde{l}$) that indicates on which energy level of $H_p$ ($H_{\tilde{p}}$) the state's weight lies. Denote by $u_i \in [0, n+1]$ the number of times that $i  \in  [1, l]$ appears in the first string; and by $v_j  \in [0, n+1]$, the number of times that $j  \in  [1, \tilde{l} ]$ appears in the second string. If two states in $\mathcal{S}_n$ correspond to the same pair 
\begin{equation}
   u  =  ( u_1, u_2, \ldots, u_l )
   \quad \text{and} \quad
   v =  ( v_1, v_2, \ldots, v_{ \tilde{l} } ),
\end{equation}
the states correspond to the same energy.
($u$ and $v$ are called $r$ and $s$ in~\cite{Janzing00}.)

A permutation $\pi_n : \mathcal{S}_n  \to  \mathcal{S}_n$ is defined in terms of the matrix assumed to map $p$ to $\tilde{p}$. Because $\pi_n$ maps each input to an output that has the same $(u, v)$, $\pi_n$ conserves energy. $\pi_n$ is applied to the probability distribution $P_n$ defined by 
$p  \otimes   g^{ \otimes n }   \otimes   \tilde{g}^{ \otimes n }   \otimes   \tilde{g}$.
A set $\mathcal{T}_n$ of typical $(u, v)$ tuples is defined in terms of $P_n$ and the limit $n \to \infty$. In this limit, Janzing \emph{et al.} show, $\pi_n$ maps $P_n$ to the distribution defined by
$g   \otimes   g^{ \otimes n }   \otimes   \tilde{g}^{ \otimes n }   \otimes   \tilde{p}$.

To adapt this Helmholtz proof to grand-potential theories, replace $p$, $\tilde{p}$, $g$, and $\tilde{g}$ with $r$, $s$, $g_R$, and $g_S$. If two states in $\mathcal{S}_n$ correspond to the same $(u, v)$, they correspond not only to the same energy, but also to the same particle number. Just as $\pi_n$ conserves energy, it conserves particle number. The rest of the proof in~\cite{Janzing00} shows that, from the equimajorization condition, an equilibrating operation can be constructed. 
\end{proof}

The proof technique used above extends from grand-potential theories to thermodynamic resource theories in which extensive-variable operators other than $H$ and $N$ commute with each other~\cite{YungerHalpern14}. Before proceeding to Proposition~\ref{prop:equimajorization}, we establish Lemma~\ref{lem:Lorenzhypotest} for convenience.

%
%
\begin{lemma*}[\ref{lem:Lorenzhypotest}]
The inverse of $\epsilon \mapsto b_\epsilon(r||g_R)$ is the piecewise linear function that connects the points $(L_R (t_k),1-t_k)$, wherein $t_k$ and $L_R (t_k)$ define the rescaled Lorenz curve for $R$:
\begin{align}
   (L_R (t_k),1-t_k)
   =   ( b_\epsilon(r||g_R),   \epsilon ).
\end{align}
\end{lemma*}

\begin{proof}
Let $\pi$ denote a permutation such that the sequence $(r_{\pi(k)}/g_{\pi(k)})_k$ is nonincreasing.  
Let $R_m   :=   \sum_{k=1}^m r_{\pi(k)}$ and $G_m   :=   \sum_{k=1}^m g_{\pi(k)}$ for all $m  \in  \{1,2,\dots, d_R\}$, wherein $d_R$ denotes the number of elements in $r$.
For $m=0$, define $R_0=G_0=0$. 
The points that define the rescaled Lorenz curve are  $(G_m,R_m)$ for $m\in \{0,1,\dots, d_R\}$.
To prove the claim, we first show that $( G_m,  1 - R_m )$ equals the 
$( b_\eps ( r || g_R ),   \eps )$ associated with an optimal hypothesis test for each $m$. 
Then, we show that optimal tests interpolate linearly between the points. 

We begin with $m=0$. The optimal test for $\eps = 1$ is  $Q=0$; thus, $b_{1}=0$.
Hence \mbox{$( b_1,  1)  =  (0, 1)  =  ( G_0,  1 - R_0 )$.}
Now, consider the hypothesis test for an $\eps_m  =  1 - R_m$ for which $m\neq 0$. Define $Q_m$ as a $d_R \times d_R$ matrix that projects onto the $m$ values of $k$ for which $r_{ \pi (k) } / g_{ \pi (k) }$ is greatest. Operation by $Q_m$ on a vector $v$ preserves the part of the support of $v$ that lies on these $m$ values of $k$ and maps all other elements of $v$ to zero.
As $ \sum_i [Q_m  r ]_i  =  R_m$, $Q_m$ is a feasible measurement element in the primal definition of 
$b_{\eps_m}$ [Eq.~\eqref{eq:opttypeIIerror}].\footnote{
$Q$ is said to be \emph{feasible} if the measurement $\{Q, \id - Q\}$ corresponds to a Type I error probability of at most $\epsilon$. The feasible measurement that minimizes the Type II error probability is \emph{optimal}.}
Therefore,
\begin{equation}   \label{eq:BetaFirstSide}
   b_{\eps_m}(r   ||   g_R )
   \leq   \sum_i [Q_m  g_R ]_i
   =   G_m.
\end{equation}   

To show that equality holds, we consider the dual problem in Eq.~\eqref{eq:htdual}. A feasible pair\footnote{
$(\mu, \tau)$ is said to be \emph{feasible} if it satisfies the constraints in Eq.~\eqref{eq:htdual}.}
$(\mu_m,\tau_m)$ that satisfies the constraint $\mu_m r - g_R  \leq  \tau_m$ is given by defining $\mu_m$ such that $r_{\pi(m+1)}/g_{\pi(m+1)}\leq 1/\mu_m<r_{\pi(m)}/g_{\pi(m)}$ and  
$\tau_m=\sum_{k=1}^m \left(\mu_m r_{\pi(k)}-g_{\pi(k)}\right)  e_k  e_k^T$,
wherein $e_k$ denotes the unit vector that has exactly one nonzero element, which corresponds to the $[\pi(k)]^{\rm th}$ energy--and--particle-number level, and the superscript $T$ denotes the transpose. 
Evaluating Eq.~\eqref{eq:htdual} shows that the two contributions dependent on $\mu_m$ cancel, by $1 - \eps_m  =  R_m$ and by the definition of $R_m$. Hence $b_{\eps_m}(r   ||   g)\geq G_m$. 
Combining this result with Ineq.~(\ref{eq:BetaFirstSide}), shows that 
\begin{align}   \label{eq:Interp0}
   b_{\eps_m}(r   ||   g_R )=G_m.
\end{align}


Now, consider a Type I error for which $\eps_m \geq \eps  \geq \eps_{m+1}$.  
Set $\lambda  \in  (0, 1)$ such that 
\begin{equation}
   1  -  \eps  =  (1-\lambda)  (1  -  \eps_m)  +  \lambda  (1 - \epsilon_{m+1}).
\end{equation}
Note that $(1  -  \eps)  =  (1  -  \eps_m)  +  \lambda r_{\pi(m+1)}$. Since 
$\eps  \mapsto b_\eps( r   ||   g_R)$ is convex, 
\begin{align}   \label{eq:Interp}
   b_{\eps}( r   ||   g_R)
   \leq (1-\lambda) b_{\eps_m} ( r   ||   g_R )   
   +   \lambda b_{\eps_{m+1}} ( r   ||   g_R ).
\end{align}

Let us show that Ineq.~\eqref{eq:Interp} holds if the inequality is reversed.
In the dual problem, if $\mu=g_{\pi(m+1)}/r_{\pi(m+1)}$ and 
$\tau   =   \sum_{k=1}^{m} \left(\mu r_{\pi(k)}-g_{\pi(k)}\right)
   e_k  e_k^T$, then
\begin{align}
   b_\eps( r   ||   g_R )   &   \geq 
   \mu  [  (1  -  \eps_m)  +\lambda r_{\pi(m+1)}  ] - \mu R_m+G_m\\
   &=\lambda g_{\pi(m+1)}+G_m\\
   &=(1-\lambda)G_m+\lambda G_{m+1}\\
   &=(1-\lambda)b_{\eps_m} ( r   ||   g_R )   +   \lambda b_{\eps_{m+1}} ( r   ||   g_R ).
         \label{eq:Interp2}
\end{align}
The final equality follows from Eq.~\eqref{eq:Interp0}. Inequalities~\eqref{eq:Interp} and~\eqref{eq:Interp2} show that interpolating linearly between 
$(G_m,  1 - R_m)$ and $(G_{m+1},  1 - R_{m+1})$ amounts to interpolating linearly between 
$( b_{\eps_m}(r || g_R),  \eps_m)$ and $( b_{\eps_{m+1}}(r || g_R),  \eps_{m+1})$.
\end{proof}

Finally, we give a mostly self-contained proof of Proposition~\ref{prop:equimajorization}.

%
%
%
%
\begin{proposition*}[\ref{prop:equimajorization}]
For any states $R$ and $S$ in the grand-potential resource theory defined by $(\beta, \mu)$, the following are equivalent:
\begin{enumerate}[(a)]
\item $R\eqmaj  S$.
\item $L_R(t)\geq L_S(t)$ for all $t\in [0,1]$.
\item $\phi_{f_a}(R)\geq \phi_{f_a}(S)$ for  every function 
$f_a(t)   =   \max  \{0,   t-a \}$ and every $a \in \mathbbm{R}$.
\item $\phi_f(R)\geq \phi_f(S)$ for all continuous convex functions $f$.
\end{enumerate}
\end{proposition*}

\begin{proof}
We will show that $(a)\Rightarrow (b)\Rightarrow (c)\Rightarrow (d)\Rightarrow (a)$.
$R$, $S$, $G_R$, and $G_S$ are defined as in the proof of Theorem~\ref{theorem:FreeOpMajor}.
\begin{itemize}

\item[$(a)\Rightarrow (b)$] 
Let $M$ denote the stochastic matrix from the equimajorization condition. Let $Q$ define the optimal test that distinguishes $s$ from $g_S$ with a Type I error probability of at most $\epsilon$. To distinguish $r$ from $g_R$, one can apply $M$ and then measure $\{ Q,  \id  -  Q \}$. 
This test might distinguish between $r$ and $g_R$ suboptimally.
Hence $b_\eps(r  ||  g_R)\leq  b_\eps(s  ||  g_S)$. Lemma~\ref{lem:Lorenzhypotest} implies (b).

\item[$(b)\Rightarrow (c)$] 
The dual formulation of $\eps\mapsto b_\epsilon(r  ||  g_R)$ can be written as 
\begin{align}
   b_\epsilon(r  ||  g_R)=\max_\mu \Big\{   (1-\eps)\mu  -
   \sum_i [\{\mu r-g_R\}_+]_i   \Big\}.
\end{align}
Hence 
$b_{1-\eps}(r  ||  g_R)$  is the Legendre transform of
\begin{equation}
   \sum_i   [\{\mu r-g_R\}_+]_i
   =  \mu  \sum_i g_i f_{1/\mu} \left(  \frac{ r_i }{ g_i }  \right)
   =  \mu \:  \phi_{f_{1/\mu}}(r,g_R), 
\end{equation}
wherein $f_a(t)   =   \max  \{0,   t-a \}$. Since $b_\eps(r  ||  g_R)\leq  b_\eps(s  ||  g_S)$, 
$\phi_{f_{1/\mu}}(r,g_R) \geq \phi_{f_{1/\mu}} (s,g_S)$. 

\item[$(c)\Rightarrow (d)$] 
In~\cite{uhlmann_ordnungsstrukturen_1978} (see also ~\cite[Lemma 1.2.5]{alberti_dissipative_1981}), Uhlmann shows that every continuous convex function $f(x)$ can be approximated to arbitrary accuracy by a linear combination, that has positive coefficients, of functions $f_a(x)$. [In Uhlmann's phrasing, a concave $f(x)$ is approximated by positive linear combinations of $-f_a(x)$]. Because  $\phi_{f_a}(R)  \geq \phi_{f_a} (S)$ for all $f_a$, $\phi_f(R)\geq \phi_f(S)$.

\item[$(d)\Rightarrow (a)$] 
Following~\cite[Theorem 1.4.4]{alberti_dissipative_1981}, we prove the contrapositive. Assume that no stochastic matrix $M$ satisfies $Mr=s$ and $Mg_R=g_S$. Since the set of stochastic matrices is convex and compact (all entries being in the unit interval), the set of vectors $Mr\oplus Mg_R$ is convex and compact. As this set does not contain $s\oplus g_S$, a hyperplane separates $s\oplus g_S$ from $\{Mr\oplus Mg_R\}$~\cite[Theorem 3.5]{tiel_convex_1984}. That is, a vector $x\oplus y$ satisfies $x\cdot s+y\cdot g_S > x\cdot Mr + y\cdot Mg_R$ for all stochastic matrices $M$. Taking the maximum over $M$ on the right-hand side (RHS) and denoting by $[g_R]_k$ the $k^{\rm th}$ element of $g_R$ gives
\begin{align}
x\cdot s+y\cdot g_S &
> \max_M\,   \left\{  x\cdot Mr + y\cdot M  g_R   \right\} \\
&=\max_M \sum_{jk}   \left(  x_jM_{jk}r_k+y_jM_{jk}  [g_R]_k  \right)  \\
&=\sum_k \max_j   \left\{   x_jr_k+y_j  [g_R]_k    \right\} . 
\end{align}
Maximizing on the LHS produces
\begin{align}
\sum_k \max_j   \left\{  x_js_k+y_j  [g_S]_k    \right\} 
   >\sum_k \max_j   \left\{  x_jr_k+y_j  [g_R]_k   \right\} .
\end{align}
But $f(s,t)=\max_j  \left\{  x_js+y_jt  \right\} $ is a convex function, so 
$\phi_f(S)\geq \phi_f(R)$. The contrapositive implies $(a)$. 
\end{itemize}
\end{proof}

%
%
%
%
\section{One-shot work-yield and work-cost proofs} \label{section:FaultyWork}

Equilibrating operations can extract work from one copy of a state $R=(r, H, N)$ and can store the work in a battery. From enough stored work, equilibrating operations can generate $R$. Each protocol can have a probability $\epsilon \in [0, 1]$ of failing to accomplish its purpose. We calculate the maximum work 
$\Wext (R)$ extractable from, and bound the least work 
$\Wform$ required to create, $R$ with error-prone protocols. 
First, we prove a helpful lemma. Anticipating applications of the lemma, we use the notation $G$ and $G'$ associated with free states. However, the lemma holds if $G$ denotes an arbitrary quantum state and $G'$ denotes an arbitrary quasiclassical state.

We use the following notation: By $e_E$, we denote a vector of the eigenvalues of the pure state $\ketbra{E}{E}$. The element associated with energy $E$ is one, and the other elements are zeroes.
By $[ v ]_i$, we denote the $i^{\rm th}$ element of any vector $v$.
If $r$ and $s$ denote equal-sized vectors, their scalar product is 
$r \cdot s  :=  \sum_i r_i s_i$.

%
%
\begin{lemma}
\label{lem:hypoconvert}
Let $R=(r,H,N)$ denote any state;
and \mbox{$( e_E,   H',  N' )$}, any pure state.
Let $G=(g,  H,  N )$ and $G'=(g',   H',  N')$. 
The optimal hypothesis test between $r$ and $g$ is related to the optimal test between 
$r  \otimes   e_E$ and $g  \otimes  g'$ by
\begin{align}
   b_\eps(r   \otimes   e_E    ||   g   \otimes   g') 
   &= (g' \cdot e_E)   \,b_\eps(r   ||   g). \label{eq:hypoconvert}
\end{align}
\end{lemma}

\begin{proof}
Consider any feasible measurement operator $Q$ in 
$b_\eps(r   \otimes   e_E   ||   g    \otimes   g')$. We can find a feasible $Q'$ that gives the same value of the objective function and that has the form 
\mbox{$Q'=\sum_{i} Q_{E'_i}   \otimes  e_{E'_i}  e_{E'_i}^T$}. Here,
\mbox{$Q_{E'_i}  =  (\id   \otimes   e_{E'_i})  Q  (\id   \otimes   e_{E'_i}^T  )$}, and the superscript $T$ denotes the transpose. Without loss of generality, we focus on $Q'$ operators that have this form. 

The constraint becomes $\sum_i  [Q_E \: r]_i  \geq 1  -  \eps$, and the objective function becomes
\begin{align}
   \sum_i   [ Q'   (g   \otimes g') ]_i   
   &=   \sum_{i}  \left\{   (g'  \cdot   e_{ E'_i })
                                    \sum_j   [Q_{E'_i}  \:   g]_j   \right\}.
\end{align}
The minimum follows from setting $Q_{E'_i}=0$ for all $E'_i   \neq E$:
\begin{align}
   b_\eps  (r   \otimes   e_E    ||   g   \otimes   g') 
   &=   \min \{  (g'  \cdot  e_E)  \sum_j  [ Q_E  \:   g ]_j  
                         \:  \Big|  \:
                       \sum_j [Q_E  \:   r ]_j   \geq   1  -  \eps,      
                       0   \leq   Q_E   \leq   \id   \}   \\
   &=   (g'  \cdot  e_E)   \,   b_\eps (r   ||   g).
\end{align}
\end{proof}

This lemma implies another lemma, associated with $\eps=0$, that will facilitate our work-bound proofs.

\begin{lemma}
Let $R$ denote any state in a grand-potential resource theory defined by $\beta$ and $\mu$. Let $W$ denote the work extractable from $R$, and let $W'$ denote the work cost of creating $R$, with error tolerance $\eps = 0$. There exist batteries $B$ such that
\begin{align}
   R  \comp   {B}_E       \succ_{\beta, \mu}       {B}_{E+W}
   \quad &\Leftrightarrow \quad 
   R   \succ_{\beta, \mu}     {B}_W
   \quad {\rm and} \label{eq:spc1}\\
   {B}_{E+W'}    \succ_{\beta, \mu}     R  \comp    {B}_{E}
   \quad &\Leftrightarrow \quad 
   {B}_{W'} \succ_{\beta, \mu} R.\label{eq:spc2}
\end{align}
\end{lemma}

\begin{proof}
Consider a battery $B_{E+W}$ that consists of two noninteracting parts (e.g., two batteries, $B_E$ and $B_W$). The total Hamiltonian is the sum of the subsystems' Hamiltonians, and the total number operator is a sum. 
Suppose that the first Hamiltonian has $d_1$ eigenvalues and the second has $d_2$.
The joint-system state vector \mbox{$ e_E  \otimes  e_W$} is an 
energy-$(E+W)$ eigenstate. The joint system's equilibrium state is the composition of the constituent systems' equilibrium states, whose state vectors we denote by $g$ and $g'$. 
That is, ${B}_{E+W}   =   {B}_E   \comp   {B}_W$. 

Applying several results, we can prove the equivalences in~\eqref{eq:spc1}:
\begin{align}
   R + B_E  \eqmaj  B_W + B_E   \quad 
   & \Leftrightarrow \quad
   L_{R + B_E}(t)  \geq  L_{B_W  +  B_E}(t)  \quad \forall t \in [0, 1] \\
   &  \Leftrightarrow \quad
   b_\eps ( r \otimes e_E  ||  g \otimes g' )
       \geq  b_\eps ( e_W \otimes e_E  ||  g \otimes g')  \\
   &  \Leftrightarrow \quad
   b_\eps (r || g)  \geq  b_\eps ( e_W || g )  \\
   &  \Leftrightarrow \quad
   L_R (t)  \geq  L_{ B_W } (t)  \quad \forall t \in [0, 1] \\
   &  \Leftrightarrow \quad
   R  \eqmaj  B_W.
\end{align}
The first equivalence follows from Proposition~\ref{prop:equimajorization}; the second, from Lemma~\ref{lem:Lorenzhypotest}; the third, from Lemma~\ref{lem:hypoconvert}; the fourth, from Lemma~\ref{lem:Lorenzhypotest}; and the fifth, from Proposition~\ref{prop:equimajorization}. Similar reasoning justifies the equivalence of Eqs.~\eqref{eq:spc2}.
\end{proof}

Having simplified the model of work, we will calculate $\Wext$.

%
%
\begin{theorem*}[\ref{thm:extractMain}]
The $\epsilon$-work value of a state $R = (r, H, N)$ associated with the free state
$G_R = (g_R, H, N)$ is
\begin{align}  \label{eq:WExtApp}
\Wext(R)   =   \tfrac1\beta\Dh^{\eps}(r   ||   g_R).
\end{align}
\end{theorem*}

\begin{proof}
In the converse part of the proof, we show that the RHS is an upper bound on the extractable work. In the direct part, we construct an equilibrating operation that attains the bound.

For the converse, define an equilibrating operation by
$\mathcal E(r   \otimes   e_E   )
   \approx_\eps   e_E   \otimes   e_W$.
The channel's output is $\eps$-close, in the $l_1$ norm, to the desired state:
\begin{equation}
   \frac{1}{2} |  \mathcal E(r   \otimes   e_E )   -   e_E   \otimes   e_W |_1   \leq   \eps.
\end{equation}
Since equilibrating operations map equilibrium states to equilibrium states, 
$\mathcal{E}(g_R   \otimes    g')=   g'$.
Using $\mathcal{E}$, we can construct a hypothesis test between $r   \otimes   e_E$ and 
$g_R   \otimes   g'$. The test consists of an application of $\mathcal{E}$ followed by an energy measurement.\footnote{
The proof does not depend on how, or whether, measurements are defined in the resource theory. Because the proof is not a protocol for extracting work, the resource-theory agent need not perform the measurement.}
If the measurement yields $E+W$, we guess that the state is $r   \otimes   e_E$. Otherwise, we guess $g_R   \otimes   g'$. By construction, the probability that we correctly guess 
$r   \otimes   e_E$ is at least $1-\eps$. The test is feasible for 
$\Dh^{\eps}(r   \otimes   e_E    ||   g_R   \otimes   g')$, and
\begin{align}
   e^{-\Dh^{\eps}(r   \otimes   e_E   ||   g_R   \otimes   g')}
   &\leq   \mathcal{E}   (g_R   \otimes   g')  \cdot  e_{E+W}   \\
&=   (g  \otimes  g')  \cdot  e_{E+W}   \\
&=\frac{e^{-\beta (E+W)}}{Z}. \label{eq:costconverse}
\end{align}
By Eq.~\eqref{eq:hypoconvert}, 
\begin{align}
e^{-\Dh^{\eps}(r   ||   g_R)}\leq e^{-\beta W},
\label{eq:basecase}
\end{align}
which is equivalent to the upper bound.

For the proof's direct part, we define the state vector
\begin{align}
   \tilde{g}'   :=   
   \frac1{1-\tfrac 1Ze^{-\beta (E+W)}}
   \left[g'-\tfrac1Ze^{-\beta (E+W)}  e_{E+W}    \right].
\end{align}
Using the optimal measurement  $Q$  in $\Dh^{\eps}(r   \otimes  e_E   ||   g_R   \otimes   g')$, we define the operation $\mathcal E$ by  
\begin{align}
   \mathcal{E}(s  \otimes  u)
   =   \left\{  1-   \sum_i   [Q   (s  \otimes  u) ]_i   \right\}   \tilde{g}'+
         \left\{ \sum_i   [Q   (s  \otimes  u) ]_i   \right\}   e_{E+W}
\end{align}
for state vectors $s$ and $u$. By construction, $\mathcal{E}(r  \otimes   e_E)$ has the form of the desired output.
Since $\mathcal E$ is an equilibrating operation, 
$\mathcal{E}(g_R   \otimes   g')=   g'$. 
This condition determines the possible values of $W$ and is equivalent to
\begin{align}
   \sum_i   [Q (  g_R   \otimes   g'  )]_i
   &=\frac{e^{-\beta (E+W)}}Z.
\end{align}
This equation is equivalent to Ineq.~\eqref{eq:costconverse}, except for containing an equality.
Therefore, free operations can distill at least the work $W$ that satisfies 
\begin{align}
e^{-\Dh^{\eps}(r   ||   g_R )}= e^{-\beta W},
\end{align}
which is equivalent to the lower bound.
\end{proof} 

Having calculated the work extractable from $R$, we bound the work cost of creating $R$ [Ineqs.~\eqref{eq:WForm}].

%
%
\begin{theorem*}[\ref{thm:extractMain}, ctd.]
The work cost of creating an $\epsilon$-approximation to a state $R$ is bounded by 
\begin{align}
 \max_{\delta\in (0,1-\eps]}  \left[
     \tfrac1\beta \Dh^{1 - \eps - \delta}(r   ||   g_R)-\tfrac1\beta\log  \left( \tfrac1\delta \right)   \right]
   \leq \Wform(R)
&\leq \tfrac1\beta\Dh^{1 - \eps}(r   ||   g_R)-\tfrac1\beta\log  \left( \tfrac {1-\eps}{\eps} \right).
\end{align}
\end{theorem*}

\begin{proof}
To derive the lower bound, we suppose that $\mathcal{E}$ is an equilibrating operation that satisfies 
$\mathcal{E}( e_{E + W} )\approx_\eps   r   \otimes  e_E$.
Using $\mathcal{E}$, we can transform the optimal dual program for $e_{E + W}$ and ${g'}$ into a feasible dual program for $\mathcal{E}( e_{E+W} )$ and $g_R   \otimes   g'$. This feasible program can be related to the hypothesis-testing entropy of $r$ relative to $g_R$.  

Consider distinguishing between $e_{E+W}$ and ${g'}$ by hypothesis test. 
Let $1$ denote the state on a one-dimensional space.
By Lemma~\ref{lem:hypoconvert}, 
\begin{align}   
e^{-\Dh^{1 - \eps}(  e_{E+W}   ||   {g'})}
&=   (g'   \cdot   e_{E+W})   {e^{-\Dh^{1 - \eps}(1   ||   1)}}\\
&=   \label{eq:PreDual}   \eps \,(g'  \cdot  e_{E+W}).
\end{align}
The dual formulation of $D^\epsilon_H$ reads,
\begin{align}  \label{eq:DualApp}
e^{-\Dh^{1 - \eps}( e_{E+W}   ||   {g'})}
   &=   \mathop{\max_{ {\mu}  \:   e_{E+W}   -   {g'}\leq {\tau} }}_{\mu,\tau\geq 0}
   \left\{  \eps  {\mu}-   \sum_i   \tau_i   \right\},
\end{align}
wherein $\tau_i$ denotes the $i^{\rm th}$ element of $\tau$. Comparing Eqs.~\eqref{eq:PreDual} and~\eqref{eq:DualApp} shows that 
${\mu}   =    g'   \cdot   e_{E+W}$ and $\tau=0$ are the optimal choices in the dual formulation. 

Acting on each side of the constraint,
$\mu \: e_{E+W}   \leq   g'$, with $\mathcal{E}$ yields
$\mu \mathcal E(  e_{E+W} )   \leq   g_R   \otimes   g'$.
Therefore, $\mu=  g'  \cdot   e_{E+W}$ and $\tau=0$ are feasible for
 $\Dh^{1 - \delta}(\mathcal E( e_{E+W} )   ||    g_R  \otimes   g')$:
\begin{align}
   e^{-\Dh^{1 - \delta}(\mathcal E( e_{E+W} )   ||   g_R   \otimes   g')}
   \geq \delta\, \frac{e^{-\beta (E+W)}}{Z}
\end{align}
for all $\delta \in [0, 1]$.
Since $\tfrac12\|   r   \otimes e_E   -\mathcal E( e_{E+W} )\|_1\leq \eps$,
\begin{equation}
   \Big|  \sum_i  [Q  \{ r   \otimes e_E  -  \mathcal E(  e_{E+W}  )  \} ]_i   \Big|
   \leq \eps
\end{equation}
for every $Q$.
Suppose $Q$ is the optimal choice in 
${\Dh^{1 - \eta}(r   \otimes   e_E   ||   g_R   \otimes g')}$, such that 
$\sum_i   [Q   \,   (r   \otimes  e_E) ]_i   =   \eta$. 
For this $Q$, $\sum_i   [Q  \:   \mathcal E( e_{E+W} ) ]_i   \geq \eta   -   \eps$. 
Therefore, 
\begin{align}
   e^{-\Dh^{1 - \eta + \eps}(\mathcal E(e_{E+W} )   ||   g_R   \otimes   g')}
   &\leq \sum_i   [Q   \,   (g_R   \otimes   g' ) ]_i\\
   &=e^{-\Dh^{1 - \eta}(r   \otimes   e_E   ||   g_R   \otimes   g')}\\
   &=\frac{e^{-\beta E}}{Z}e^{-\Dh^{1 - \eta}(r   ||   g_R )}.
\end{align}
If $\eta=\eps+\delta$,
\begin{align}
\delta e^{-\beta W}\leq {e^{-\Dh^{1 - \eps - \delta}(r   ||   g_R )}}.
\end{align}
We have lower-bounded $\Wform (R)$ for every $\delta\in(0,1-\eps]$.  
The tightest of these bounds follows from a maximization over $\delta$.

To derive upper bound, we construct an equilibrating operation that maps 
$e_{E+W}$ to $\tilde{r}   \otimes e_E$, wherein $\tilde{r}   \approx_\eps r$, for a suitably chosen value of $W$. The work system is approximated, whereas the battery is not. The associated work-cost bound may be suboptimal. 

By Condition (c) of Proposition~\ref{prop:equimajorization}, such an equilibrating operation exists if and only if 
\begin{equation}   \label{eq:MonoK}
   K_{\rm in}(a)\geq K_{\rm out}(a)
\end{equation}
for all $a   \in \mathbbm{R}$ and for $K_R(a)$ defined as follows. In Proposition~\ref{prop:equimajorization}, the function 
\mbox{$f_a(t)  :=  \max \{0,  t - a \}$} appears in
\begin{align}
   \phi_{f_a}(R)  &:=   \sum_{i=1}^{ d_R }   g_i  f_a \left(  \frac{r_i}{g_i}  \right)  \\
   & =   \sum_{i=1}^{ d_R }   g_i  \max \left\{0,  \frac{r_i}{g_i}  -  a \right\}  \\
   & =   \sum_{i=1}^{ d_R }   \max \{0,  r_i  -  g_i a \}.
\end{align}
To simplify notation, we will relabel this sum as $K_R(a)$.
Because $r$ and $g$ are normalized, $K_R(a) = (1-a)$ if $a\leq0$. 
We can rewrite the LHS of Ineq.~\eqref{eq:MonoK} as
\begin{align}
   K_{\rm in}(a)  
   =   \left(1-  a  \frac {e^{-\beta(E+W)}}{Z}\right)_+.
\end{align}
This function is linear and satisfies $K_{\rm in}(0)=1$ and $K_{\rm in}(Ze^{\beta(E+W)})=0$. We can rewrite the RHS of Ineq.~\eqref{eq:MonoK} as
\begin{align}
K_{\rm out}(a)=\sum_i\left(\tilde{r}_i   -   g_i   a \frac {e^{-\beta E}}{Z}\right)_+.
\label{eq:Kout}
\end{align}
Just as for the input state, $K_{\rm out}(0)=1$. As a sum of convex functions, $K_{\rm out}(a)$ is convex. Thus, the condition $K_{\rm in}(a)\geq K_{\rm out}(a)$ for all $t\in \mathbbm R$ reduces to 
\begin{align}
K_{\rm out}(Ze^{\beta (E+W)})= 0.\label{eq:Kcondition}
\end{align}

Let us find a value of $W$ for which the transformation is possible. First, we construct a suitable $\tilde{r}$ from the dual form of $\Dh^{1 - \eps}(r   ||   g_R)$. 
Suppose that $\mu$ and $\tau$ are the optimal choices, so that 
\begin{align}
e^{-\Dh^{1 - \eps}(r   ||   g_R)}=  \eps   \mu-   \sum_i  \tau_i
\end{align}
and $\mu  r   -   g_R   \leq   \tau$. We define 
$r'   =   T   r T^\dagger$ for $T   =   g_R^{\nicefrac12}(g_R   +   \tau)^{-\nicefrac12}$, using the pseudoinverse (the inverse on the support). These definitions satisfy $\mu   r'   \leq   g_R$.  Let us bound the trace distance
\begin{align}
\tfrac12\|r   -   r'\|_1&=\tfrac12\sum_k |r_k   -   r'_k|\\
&=\tfrac12\sum_k|   r_k  -   g_k   r_k   (g_k+   \tau_k)^{-1}|\\
&=\tfrac12\sum_k   r_k|1-   g_k( g_k   +   \tau_k)^{-1}|\\
&=\tfrac12\sum_k   r_k\frac{\tau_k}{ g_k+   \tau_k}\\
&\leq \tfrac12\sum_k\tfrac1\mu \tau_k\\
&=\frac{ \sum_i   \tau_i  }{2\mu}\\
&\leq \tfrac\eps 2.
\end{align}
To derive the first inequality, we used the inequality $\mu  r  \leq  g_R  +  \tau$; to derive the second, $\eps   \mu-  \sum_i  \tau_i  =e^{-\Dh^{1 - \eps}(r   ||   g_R)}\geq 0$. 
Let $\tilde{r}   =   r'/  \sum_i r'_i$. By the Triangle Inequality, 
\begin{align}
   \tfrac12\|  r   -   \tilde{r}   \|_1&\leq \tfrac12\|  r   -   r'\|_1+\tfrac12\|\tilde{r}   -   r'\|_1\\
&\leq \tfrac\eps2+\tfrac12\left(\frac1{\sum_i  r'_i  }-1\right)\|  r'\|_1\\
&=  \tfrac\eps2+\tfrac12\left(1-   \sum_i  r'_i  \right)\\
&\leq \eps.
\end{align}
The fourth inequality follows from $\sum_i  r_i  -   \sum_j   r'_j   \leq \|  r   -   r'\|_1$. We have constructed a state $\tilde{r}   \approx_\eps   r$.

Moreover, $\mu \left( \sum_i  r'_i  \right)  \tilde{r}   \leq  g$. Applying this inequality to \eqref{eq:Kout} yields
\begin{align}
   K_{\rm out}(a) 
   &\leq \sum_i g_i
       \left(\frac1{\mu   \sum_i  r'_i  }-  a \frac{e^{-\beta E}}{Z}\right)_+\\
   &=\left(\frac1{\mu   \sum_i  r'_i }-  a \frac{e^{-\beta E}}{Z}\right)_+.
\end{align}
Hence $K_{\rm out}(Ze^{\beta(E+W)})$ satisfies
\begin{align}
   K_{\rm out}(Ze^{\beta(E+W)}) &\leq \left(\frac1{\mu  \sum_i  r'_i  }-e^{\beta W}\right)_+\\
   &\leq \left(\frac1{\mu(1-\eps)}-e^{\beta W}\right)_+\\
   &\leq \left(\frac{\eps}{1-\eps} e^{\Dh^{1 - \eps}(r   ||   g)}-e^{\beta W}\right)_+.
 \label{eq:Koutupper}
\end{align}
The second inequality follows from $\sum  r'_i   \geq 1-\eps$; the third, from 
$\eps  \mu   \geq   \eps  \mu-   \sum_i   \tau_i
=e^{-\Dh^{1 - \eps}(r   ||   g_R)}$. 
We can satisfy \eqref{eq:Kcondition} by choosing $W$ such that 
\begin{align}e^{\beta W} = \frac{\eps}{1-\eps} e^{\Dh^{1 - \eps}(r  ||   g_R)}.
\end{align}
Since $\Wform (R)   \leq W$, the upper bound follows directly.
\end{proof}

One can show that the upper bound is a nonnegative quantity, using 
$e^{-\Dh^{1 - \eps}(r   ||   g_R)}\leq   \eps$. 
(This inequality follows from the choice $Q=\eps \id$.)
Because $\eps\leq \frac \eps{1-\eps}$, 
$e^{-\Dh^{1  -  \eps}(r   ||   g_R)}\leq \frac \eps{1-\eps}$.
The latter implies the upper bound's nonnegativity. 
The lower bound is nonnegative in all the numerical examples we tested.

%
%
%
%
\section{Comparison of one-shot work yield and work cost}   \label{section:SecondOrderAsymp}

We will use \emph{second-order asymptotics} to show that $\Wext(R^{\otimes n})$ tends to differ from the bounds on $\Wform( R^{\otimes n} )$, and that the bounds lie arbitrarily close together, as the thermodynamic limit is approached. Consider distilling work from, or creating, $n$ copies $R^{\otimes n}$ of $R=(r,H,N)$. The work involved depends on the \emph{normal approximation} to the hypothesis-testing relative entropy \cite[Theorem 5]{li_second-order_2014} (see also~\cite{strassen_asymptotische_1962,polyanskiy_channel_2010}):\footnote{
The quantum version of Eq.~\eqref{eq:DattaLed} appears in~\cite{li_second-order_2014,tomamichel_hierarchy_2013}, but we have specialized to commuting density operators.}
\begin{equation}   \label{eq:DattaLed}
   D^\epsilon_{\rm H} (   r^{\otimes n}  ||  g_R^{\otimes n}  )
   =   n D( r || g_R )
   + \sqrt{n}  \; s(r || g_R )   \: \Phi^{-1} ( \epsilon )
   +   O( \log n ),
\end{equation}
wherein 
$g_R$ denotes the state vector of the equilibrium state associated with $R$,
the square-root of the relative entropy variance is
\begin{equation}
   s( r || g_R)   
   :=   \sqrt{ V( r  ||  g_R ) }
   =   \sqrt{  {\rm Tr} ( r [ \log r   -   \log g_R ]^2 )
                     - D ( r  ||  g_R )^2  },
\end{equation}
and the inverse error function is
\begin{equation}   \label{eq:Phi1}
   \Phi^{-1}( \epsilon )
   :=   \sup   \left\{   z   \in   \mathbbm{R}   \; \big| \;
   \frac{1}{ \sqrt{2 \pi } }  
        \int_{ - \infty }^z 
        e^{ - \frac{1}{2}   t^2   }   dt
    \leq \epsilon   \right\}.
\end{equation}
Equation~(\ref{eq:Phi1}) admits of the following interpretation.
Suppose that, if a hypothesis test is performed on the null-hypothesis state $R$, the outcome is distributed normally. The probability that a Type I error occurs equals $\Phi^{-1}(\epsilon)$.
Let us apply Eq.~(\ref{eq:DattaLed}) to Eq.~(\ref{eq:WExt}) and to Ineqs.~(\ref{eq:WForm}).

To characterize the latter expressions' approach toward $D(r || g_R)$, we evaluate the normalized differences between each $W^{\epsilon, \beta} \big(R^{\otimes n}  \big)$ and 
$\frac{1}{\beta} D(r || g_R)$, assuming $n$ is large. The actual distillable work differs from the asymptotic distillable work as
\begin{align}
   \label{eq:WCompare1}
   \lim_{n\rightarrow \infty}\frac{1}{ \sqrt{n} } \left[
      n \: \frac{1}{ \beta }   D( r || g_R )  -
      \Wext \big(  R^{\otimes n}  \big)
   \right]  
    &= 
   \lim_{n\rightarrow \infty}\frac{1}{ \beta }   \left[
      - s ( R || g_R )   \Phi^{-1}(\epsilon)
      - \frac{ O ( \log n ) }{ \sqrt{n} }   \right] 
    \\
   &   \label{eq:GainPhi}   =  \frac{1}{ \beta }   s ( R || g_R )   \Phi^{-1}(1 - \epsilon).  
\end{align}
The final equation follows from 
$\Phi^{-1} ( \epsilon )  =  - \Phi^{-1} (1 - \epsilon)$. 
If the Type I error probability is small ($\epsilon < \frac12$), Eq.~(\ref{eq:GainPhi}) is positive because the work distilled at the optimal asymptotic efficiency exceeds the work distilled at any sub-asymptotic efficiency [i.e., Eq.~(\ref{eq:WCompare1}) is positive].

The lower work-cost bound differs from the asymptotic cost by
\begin{align}
   \lim_{n\rightarrow \infty}\frac{1}{ \sqrt{n} }   \left[
      \Wform \big(  R^{\otimes n}  \big)v
      -   n \frac{1}{\beta}   D(r || g_R) \right]
   &  \geq\lim_{n\rightarrow \infty}
   \frac{1}{\beta}   \max_{ \delta   \in   (0,  1 - \epsilon ] }
   \left[   s( r || g_R ) \Phi^{-1} (1 - \epsilon - \delta )
             -   \frac{   \log \frac{1}{\delta}    }{   \sqrt{n}   }   \right]
   \nonumber  \\
   &=    \frac{1}{\beta}   \max_{ \delta   \in   (0,  1 - \epsilon ] }
   s( r || g_R ) \Phi^{-1} (1 - \epsilon - \delta )\\
   &= \label{eq:CostPhiLow} \frac{1}{ \beta } s(r || g_R)   \Phi^{-1} (1  -  \epsilon )  .
\end{align}
The first equality holds if $\delta$ grows more slowly than $e^{ \sqrt{n} }$. The last equality holds since, by the definition and monotonicity of $\Phi^{-1}$, the least possible $\delta$-value maximizes $\Phi^{-1}(1-\eps-\delta)$. 
In the limit, this difference arising from the lower bound matches that of the upper bound, as the work cost's upper bound differs from the asymptotic work cost as
\begin{align}
   \lim_{n\rightarrow \infty}\frac{1}{ \sqrt{n}  }   \left[
      \Wform \big(  R^{\otimes n}  \big)
      -   n  \frac{1}{\beta}  D(r|| g_R ) \right]
   & \leq
   \lim_{n\rightarrow \infty}\frac{1}{ \beta }   \left[
      s(r || g_R)   \Phi^{-1} (1  -  \epsilon )
      -   \frac{1}{ \sqrt{n} }   \log \left( \frac{1  -  \epsilon}{ \epsilon }  \right)   \right]
   \nonumber \\   & \label{eq:CostPhiHigh}= \frac{1}{ \beta } s(r || g_R)   \Phi^{-1} (1  -  \epsilon ) .
\end{align}
The final equality holds if $\frac{1  -  \epsilon}{ \epsilon }$ grows more slowly than $e^{ \sqrt{n} }$.

Let us compare these normalized work differences. 
If $n$ is large, the work-cost bounds exceed the optimal asymptotic work cost 
$\frac{1}{\beta}  D(r || g_R)$ by an amount proportional to $\sqrt{n}$. In contrast, 
$\frac{1}{\beta}  D(r || g_R)$ exceeds the work gain by an amount proportional to $\sqrt{n}$. The work gain and work cost differ in general, unlike in the thermodynamic limit, as in~\cite{FundLimits2,BrandaoHNOW13}. As creating $R^{\otimes n}$ requires more work than can be extracted from $R^{\otimes n}$, resources degrade.

\end{appendices}

%
%
%
%
\bibliography{MergedBib}

\end{document}